\documentclass[11pt]{article}

\usepackage{url}
\usepackage{xcolor}
\usepackage{color}
\usepackage{graphicx}
\usepackage{amssymb}
\usepackage{amsmath,amsthm,paralist}
\usepackage{amsbsy}
\usepackage{algorithm,algorithmic}
\usepackage{afterpage}
\usepackage{multirow}
\usepackage{caption}
\usepackage{bbm}
\usepackage[margin=1in]{geometry}

\usepackage{epstopdf}

\usepackage{comment}

\newcommand{\edit}[1]{{#1}} 
\newcommand{\editt}[1]{{#1}} 

\allowdisplaybreaks

\newtheorem{thm}{Theorem}

\newtheorem{prop}{Proposition}

\newtheorem{lemma}{Lemma}
\newtheorem{definition}{Definition}
\newtheorem{remark}{Remark}

\newcommand{\W}{\omega}

\newcommand{\bitem}{\begin{itemize}}
\newcommand{\eitem}{\end{itemize}}

\newcommand{\supp}{\mathrm{supp}}
\newcommand{\beqn}{\begin{equation}}
\newcommand{\eeqn}{\end{equation}}
\newcommand{\balign}{\begin{align}}
\newcommand{\ealign}{\end{align}}

\def\w{{\mathbf w}}

\def \R {\mathbb{R}}

\begin{document}

\title{Weighted $\ell_1$-Minimization for Sparse Recovery under Arbitrary Prior Information}

\author{Deanna Needell\thanks{D.~Needell is with the Department of Mathematical Sciences, Claremont McKenna College, Claremont, CA, 91711, USA (email: dneedell@cmc.edu).},  
Rayan Saab\thanks{R.~Saab is with the Department of Mathematics, University of California, San Diego, La Jolla, CA, 92093, USA (email: rsaab@ucsd.edu).},  and  
Tina Woolf\thanks{T.~Woolf is with the Institute of Mathematical Sciences, Claremont Graduate University, Claremont, CA, 91711, USA (email: tina.woolf@cgu.edu).}}

\maketitle

\begin{abstract}
{Weighted $\ell_1$-minimization has been studied as a technique for the reconstruction of a sparse signal from compressively sampled measurements when prior information about the signal, in the form of a support estimate, is available. In this work, we study the recovery conditions and the associated recovery guarantees of weighted $\ell_1$-minimization when arbitrarily many distinct weights are permitted. For example, such a setup might be used when one has multiple estimates for the support of a signal, and these estimates have varying degrees of accuracy. Our analysis yields an extension to existing works that assume only a single \editt{support estimate set upon which a constant weight is applied}. We include numerical experiments,  with both synthetic signals and real video data, that demonstrate the benefits of allowing non-uniform weights in the reconstruction procedure.}
\end{abstract}

\begin{center}
\bf{Index Terms}\\
\end{center}
\vspace{-2mm}
\hspace{9mm}  Compressed sensing, weighted $\ell_1$-minimization, restricted isometry property

\section{Introduction}
Compressed sensing is a recently developed paradigm for the effective acquisition of sparse signals via few nonadaptive, linear measurements (e.g., see \cite{Donoh_Compressed,CandeRT_Stable,CandeRT_Robust}). Specifically, we acquire
\begin{align}\label{sensing model}
y = Ax + z,
\end{align}
where $x\in\mathbb{R}^n$ is an unknown $s$-sparse signal (that is, $\|x \|_0 =|\{i: x_i \neq 0 \}| \leq s$ such that $s\ll n$), $A$ is a known $m\times n$ measurement matrix, and $z\in\mathbb{R}^m$ is measurement noise where $\|z \|_2 \leq \epsilon$. Given knowledge of the noisy observations $y\in\mathbb{R}^m$ and the matrix $A$, we are interested in obtaining an estimate of $x$ in the case where $m$ is smaller than $n$. A common approach to this task is solving the convex optimization program
\begin{align}\label{L1 min}
\min_{\tilde{x}} \|\tilde{x} \|_1 \mbox{ subject to } \|A\tilde{x}-y\|_2 \leq \epsilon,
\end{align}
where $\|\tilde{x} \|_1 = \sum_{i=1}^n |x_i|$ in the objective promotes sparsity and the constraint ensures fidelity to the model. The problem (\ref{L1 min}) can be recast as a linear program and can thus be solved efficiently in polynomial-time complexity \cite{BoydV_Convex}. It was shown in \cite{CandeRT_Stable} that if $A$ satisfies a certain restricted isometry property (which holds for certain random matrices with high probability whenever $m \gtrsim s\log(n/s)$) then the $\ell_1$-minimization in (\ref{L1 min}) can stably and robustly recover $x$ from the potentially noisy measurements $y$. 

In using (\ref{L1 min}) for reconstructing $x$, all the indices $i \in \{1,...,n\}$ are treated equally. That is,  \eqref{L1 min} does not make use of any information or assumptions on the support of $x$. 
In many applications, however, additional prior knowledge about the signal support may be available. For instance, in the acquisition of video or audio signals there is often high correlation from one frame to the next, suggesting that the information learned in the previous frames should be exploited in the acquisition of subsequent frames. In applications dealing with natural images, the signal may often be expected to follow a certain structured sparsity model which can be utilized by the recovery algorithm (e.g., \cite{BaranCDH_Model}). A \textit{weighted} version of the $\ell_1$-minimization in (\ref{L1 min}) has been studied to exploit prior information during the signal reconstruction (see Section \ref{sec::prior info} for a discussion of prior work). More precisely, suppose $\widetilde{T}\subset \{1,\dots,n \}$ is a prior support estimate. Then the weighted $\ell_1$-minimization problem can be formulated as
\begin{equation} \label{eqn::WL1}
\min_{\tilde{x}} \|\tilde{x}\|_{1,\overrightarrow{\w}} \mbox{ subject to } \|A\tilde{x}-y\|_2 \leq \epsilon, \mbox{ where } \|x\|_{1,\overrightarrow{\w}} = \sum_{i=1}^n \w_i|x_i| \mbox{ and } \w_i = 
  \begin{cases} 
      w\in[0,1] & i\in\widetilde{T} \\
      1 & i\in\widetilde{T}^c
   \end{cases}
   .
\end{equation}
Here, $\overrightarrow{\w} = [\w_1, \w_2,\dots, \w_n]$ is the vector of weights (with the only distinct weights being either $w$ or 1). Of course, (\ref{eqn::WL1}) reduces to (\ref{L1 min}) when $w=1$. The intuition behind this formulation is that selecting a weight $w$ less than 1 will encourage nonzero entries on $\widetilde{T}$ in the minimizer of (\ref{eqn::WL1}). If $\widetilde{T}$ is an accurate estimate of the true support of $x$, this weighted procedure has been shown to outperform (\ref{L1 min}), see e.g. \cite{SaabM_weighted} and references therein. 

A limitation of (\ref{eqn::WL1}) and much of the corresponding theory is that the weight $w$ is the same on the entire set $\widetilde{T}$ \editt{(as such, throughout we will refer to (\ref{eqn::WL1}) as the weighted $\ell_1$-minimization formulation with a single weight, where the weight is $w$)}. It is plausible that a practitioner may not have the same level of confidence on the entire set $\widetilde{T}$ depending on the type of prior information available, which suggests that allowing distinct weights on different pieces of $\widetilde{T}$ may be desirable.  In other settings, a statistical prior on the signal may be known, providing probabilities on each entry being in the signal support; this information should also be leveraged by using non-uniform weights.  Although the formulation and implementation of (\ref{eqn::WL1}) can be easily modified to capture this feature, the modification to the theoretical analysis is less straight-forward. In this work, we provide a generalized theory studying the recovery conditions and  error guarantees associated with weighted $\ell_1$-minimization, when \textit{arbitrary} weight assignments are permitted.  

\subsection{Prior Work}\label{sec::prior info}
Compressed sensing in the presence of prior information has previously been studied under various models. For example, the paper \cite{ChenTL_PICCS08}, which we believe to be the first of its kind, considers prior information in the form of a similar signal known beforehand. The authors propose solving a minimization problem similar to (\ref{L1 min}), but where the function being minimized includes two terms: one for measuring the sparsity of the recovered signal and the other for measuring the sparsity of the difference between the recovered signal and the prior known signal. Along the same line, the authors of \cite{MotaDR14} study a modification of the $\ell_1$-minimization (\ref{L1 min}) by minimizing either $\|\tilde{x} \|_1 + \beta\|\tilde{x}-x' \|_1$ or $\|\tilde{x} \|_1 + \frac{\beta}{2}\|\tilde{x}-x' \|_2^2$, where $x'$ denotes the similar, previously known signal, and $\beta>0$ is a parameter that establishes the tradeoff between the signal sparsity and the fidelity to the prior signal $x'$. Similarly, in the context of longitudinal Magnetic Resonance Imaging (MRI) \cite{Eldar_MRI} assumes knowledge of a previous MRI scan. Their proposed minimization is similar to that studied in \cite{ChenTL_PICCS08} and \cite{MotaDR14}, but the authors also introduce weights in both terms of the objective function. Their final proposed scheme (which is not studied analytically) solves the optimization iteratively, with the possibility of adaptively selecting measurements at each iteration. 

Another common type of prior information studied is in the form of a support estimate. The paper \cite{VaswaL_Modified} assumes a support estimate $\widetilde{T}$, and then minimizes $\|\tilde{x}_{\widetilde{T}^c} \|_1$, where $\tilde{x}_{\widetilde{T}^c}$ denotes the signal $\tilde{x}$ restricted to the indices on the complement of $\widetilde{T}$ with zeros elsewhere, which is equivalent to (\ref{eqn::WL1}) with $w=0$. In the noise-free setting, the authors show that when $\widetilde{T}$ is a reasonable estimate of the true support, the recovery conditions are milder than those needed for classical compressed sensing. The work \cite{VonBorriesMP_priorInfo07} considers prior information on the support of the discrete Fourier transform of the signal. Using a similar method as in \cite{VaswaL_Modified} with noise-free measurements, the authors show empirically that a reduced number of measurements is required for recovery. Assuming partially known support information, the authors of \cite{LiangY_dynamic10} also propose a modification to $\ell_1$-minimization similar to \cite{VaswaL_Modified} and provide experimental results in the application of dynamic MRI. The authors of \cite{Saab_weightedL1RIP} analyze the weighted $\ell_1$-minimization problem (\ref{eqn::WL1}) under a \textit{restricted isometry property} on the sensing matrix (see Section \ref{sec::existing results}), thus generalizing the results of \cite{CandeRT_Stable} to the weighted case. The authors show that if at least 50\% of the estimated support $\widetilde{T}$ is correct, then weighted $\ell_1$-minimization is stable and robust under weaker sufficient conditions than those for standard $\ell_1$-minimization. As an extension to \cite{Saab_weightedL1RIP}, an analysis of weighted $\ell_1$-minimization with multiple support estimate sets with distinct weights, which is the focus of this paper, is provided in \cite{MansourY_MultiWeightedL1}; however, we will demonstrate in Section \ref{sec::weighted l1 theory} that the main result of \cite{MansourY_MultiWeightedL1} only provides a sub-optimal generalization to \cite{Saab_weightedL1RIP}, which we remedy here. The analysis in \cite{SaabM_weighted} studies the noise-free weighted $\ell_1$-minimization problem in the presence of a support estimate $\widetilde{T}$, but, in contrast to \cite{Saab_weightedL1RIP}, the sensing matrix is assumed to possess a \textit{null space property} (see \cite{SaabM_weighted} for details). Exact recovery conditions of the noise-free version of (\ref{eqn::WL1}) under a null space property and the restricted isometry property are also studied in \cite{ZhouXWK_weighted13}. 
\edit{The work \cite{BahW_weighted15} studies the minimal number of Gaussian measurements required to achieve robust recovery via weighted $\ell_1$-minimization using the tools of weighted sparsity and weighted null space property.} 
In a separate but related direction, modifications of \textit{greedy} algorithms for compressed sensing have also been studied under the assumption of a partially known support. The work \cite{Carillo_partiallyKnownIHT} proposes a modification of the IHT \cite{BlumeD_Iterative} algorithm to incorporate partially known support information, with a theoretical bound on the reconstruction error provided. Similarly, \cite{Carillo_partiallyKnown} proposes a modification of the OMP \cite{TroppG_Signal}, CoSaMP \cite{NeedeT_CoSaMP}, and re-weighted least squares \cite{CarilloB_reweightedLS09} algorithms, and \cite{NorthOneBit15} proposes a modification of the BIHT \cite{JacquLBB_Robust} algorithm for one-bit compressed sensing to incorporate partially known support information.

As an alternative to prior support information, the papers \cite{KhajXAH_weightedl1_prior,KhajXAH_weighted11,KrishOH_Simpler12} assume a non-uniform sparsity model and analyze the noise-free weighted $\ell_1$-minimization while allowing for non-uniform weights. Specifically, the authors consider a model where the entries of the unknown signal fall into two (or more) sets, each with a different probability of being nonzero. Indices within the same set would then get assigned the same weight. The study of this type of model is further generalized in \cite{MisraP_nonuniform15}. The prior information studied in \cite{ScarlettED_priorInfo13} is also in the form of probabilities that each entry of the signal is nonzero (i.e., a prior distribution). They study information-theoretic limits on the number of (noisy) measurements needed to recover the support set exactly, and show that significantly fewer measurements can be used if the prior distribution is sufficiently non-uniform. 

Other relevant works include \cite{KhajXAH_weighted11,MisraP_nonuniform15,DiazJRM_prior16}, which propose methods for determining \edit{good} weights.  We also mention that perhaps the first study of a weighted $\ell_1$-minimization approach was in \cite{CandeWB_Enhancing}; the algorithm proposed there does not assume any prior information, but consists of solving a sequence of weighted $\ell_1$-minimization problems where the weights for the next iteration are determined from the solution of the previous iteration.  Precise recovery guarantees for this scheme have remained elusive, but it is possible 
the analysis we present here may lend further insight into this related approach.

\subsection{Contribution}
We study the weighted $\ell_1$-minimization problem (\ref{eqn::WL1}) in its full generality. 
\editt{Under a restricted isometry property on the sensing matrix,} we derive stability and robustness guarantees for weighted $\ell_1$-minimization \editt{with completely arbitrary weights} that generalize the results of \cite{Saab_weightedL1RIP}, further generalize the results of \cite{CandeRT_Stable}, and improve upon the results of \cite{MansourY_MultiWeightedL1}, thus providing an extension to the existing literature. 
Our main technical result is Theorem \ref{thm::WL1 Nsupp}, and in the ensuing discussion in Section \ref{sec::weighted l1 theory}, we compare the theoretical results associated with using a single weight to the general ones derived in this paper. We highlight scenarios under which the sufficient conditions associated with our ``multi-weight" scenario are weaker than those associated with applying a single weight to the support estimate in \eqref{eqn::WL1}, suggesting a practical benefit to the use of multiple weights. 
Indeed, we demonstrate using extensive numerical experiments that allowing arbitrary weights can often outperform weighted $\ell_1$-minimization with a constant weight. 

\subsection{Organization}
The remainder of this paper is organized as follows. In Section \ref{sec::existing results} we recall the results on weighted $\ell_1$-minimization (\ref{eqn::WL1}) when a constant weight $w$ is used for signal reconstruction in the presence of a support estimate $\widetilde{T}$, which reduce to the existing results for compressed sensing when $\ell_1$-minimization (\ref{L1 min}) is used for signal reconstruction without any prior information. In Section \ref{sec::weighted l1 theory} we present our main theoretical results along with their proofs, thus providing a generalized and improved theory of weighted $\ell_1$-minimization when non-uniform weights are permitted.
Numerical experiments involving synthetic and real signals are provided in Section \ref{sec::weighted experiments}. We conclude with a brief discussion in Section \ref{sec::discussion}.

\section{Existing Results for $\ell_1$-Minimization and Weighted $\ell_1$-Minimization with a Single Weight}\label{sec::existing results}

As mentioned earlier, in \cite{CandeRT_Stable} it was shown that the $\ell_1$-minimization problem (\ref{L1 min}), which utilizes no prior information about the signal, can stably and robustly recover $x$ from the noisy measurements $y=Ax+z$ as long as $A$ satisfies a particular property. The now well-known condition required on $A$ is termed the \textit{restricted isometry property} (RIP), and is defined in Definition \ref{def::rip} below. 

\begin{definition}\label{def::rip}
An $m\times n$ matrix $A$ is said to possess the RIP with $s$-\textit{restricted isometry constant} $\delta_s<1$ if $\delta_s$ is the smallest positive number such that
\begin{align}
(1-\delta_s)\|x\|_2^2 \leq \|Ax \|_2^2 \leq (1+\delta_s)\|x \|_2^2
\end{align}
holds for all $s$-sparse vectors $x$. 
\end{definition}
\noindent Matrices constructed with independent and identically distributed standard Gaussian entries or rows subsampled from the discrete Fourier transform matrix are examples of matrices known to satisfy the RIP when $m \gtrsim s\log^a(n/s)$, for some constant $a\geq 1$ \cite{CandeT_Near,RudelV_Sparse}.

The main result of \cite{Saab_weightedL1RIP} generalizes the recovery condition from \cite{CandeRT_Stable} to the weighted $\ell_1$-minimization problem (\ref{eqn::WL1}) where the constant weight $w$ is applied on all of $\widetilde{T}$. Theorem \ref{thm::WL1} below states the main result of \cite{Saab_weightedL1RIP}, which reduces to the result from \cite{CandeRT_Stable} when $w=1$, or when $\widetilde{T}$ is empty (see \cite{Saab_weightedL1RIP} for details).

\begin{thm}\emph{(Friedlander et al. \cite{Saab_weightedL1RIP})}\label{thm::WL1}
Let $x\in\R^n$ and let $x_s$ denote its best $s$-term approximation, supported on $T_0$. Let $\widetilde{T}\subset \{1,\dots,n\}$  be an arbitrary set and define $\rho$ and $\alpha$ such that $|\widetilde{T}| = \rho s$ and $\frac{|\widetilde{T}\cap T_0|}{|\widetilde{T}|} = \alpha$. Suppose that there exists an $a\in \frac{1}{s}\mathbb{Z}$ with $a\geq (1-\alpha)\rho$, $a>1$, and the measurement matrix $A$ has RIP with
\begin{align}\label{WL1 RIP condition}
\delta_{as} + \frac{a}{b^2}\delta_{(a+1)s} < \frac{a}{b^2} - 1
\end{align}
where
\begin{align}\label{b for 1 support}
b = b(w,\rho,\alpha) :=  w + (1-w)\sqrt{1+\rho-2\alpha\rho}
\end{align}
for some given $0\leq w \leq 1$. Then the solution $\hat{x}$ to (\ref{eqn::WL1}) obeys
\begin{align}
\|\hat{x} - x \|_2 \leq C_0\epsilon + C_1s^{-1/2}\left(w\|x-x_s \|_1 + (1-w)\|x_{\widetilde{T}^c\cap T_0^c} \|_1 \right),
\end{align}
where $C_0$ and $C_1$ are well-behaved constants that depend on the measurement matrix $A$, the weight $w$, and the parameters $\rho$ and $\alpha$.
\end{thm}
\begin{remark}
The constants $C_0$ and $C_1$ are explicitly given by
\begin{align}\label{eqn::WL1 2supp C}
C_0 = \frac{2(1+\frac{b}{\sqrt{a}})}{\sqrt{1-\delta_{(a+1)s}} - \frac{b}{\sqrt{a}}\sqrt{1+\delta_{as}}}, \,\,\,\,
C_1 = \frac{2a^{-1/2}(\sqrt{1-\delta_{(a+1)s}}+\sqrt{1+\delta_{as}})}{\sqrt{1-\delta_{(a+1)s}} - \frac{b}{\sqrt{a}}\sqrt{1+\delta_{as}}}.
\end{align}
\end{remark}
\begin{remark}
Note that the classical result of \cite{CandeRT_Stable} for un-weighted $\ell_1$-minimization (with $w=1$ and $b=1$) is proved using the condition (\ref{WL1 RIP condition}) with $a=3$, yielding the requirement $\delta_{3s} + 3\delta_{4s} < 2$.  Thus, Theorem \ref{thm::WL1} requires a \textit{weaker} RIP assumption than the classical un-weighted case for accurate support estimates (e.g. if $\rho = 1$ and $\alpha > 1/2$).  Note also that the classical condition on the RIP constant has been improved several times \cite{Cande_Restricted,AnderssonS_unifRecovery,CaiWX_New,CaiWX_Shifting,Foucart_note,FoucartL_Sparsest,MoL_New}; although a version of Theorem \ref{thm::WL1} and the main results of this paper can likely be extended to the theory of these works, we do not pursue such refinements here.
\end{remark}
\begin{remark}
Since $\delta_{as}<\delta_{(a+1)s}$, a sufficient condition for (\ref{WL1 RIP condition}) to hold is 
\begin{align}\label{WL1 RIP condition sufficient}
\delta_{(a+1)s} < \frac{a-b^2}{a+b^2} := \delta^b.
\end{align}
\end{remark}

\section{Weighted $\ell_1$-Minimization with Non-uniform Weights}\label{sec::weighted l1 theory}
In this section, we present our main results for generalizing the weighted $\ell_1$-minimization theory of \cite{Saab_weightedL1RIP}, and improving the theory of \cite{MansourY_MultiWeightedL1}, to allow for arbitrary weight assignments. 
Our main theorem is provided is Section \ref{sec::weighted l1 N}, and Section \ref{sec::proof} details the proof of this result. 

\subsection{Weighted $\ell_1$-Minimization with $N$ Distinct Weights} \label{sec::weighted l1 N}
We consider weighted $\ell_1$-minimization with $N$ distinct weights, where $1\leq N \leq n$. To that end, suppose we have $N$ disjoint support estimates $\widetilde{T}_i\subset \{1,\dots,n\}$, $i=1,\dots,N$, where $|\widetilde{T}_i| = \rho_is$. Define the accuracy of the support estimates to be $\alpha_i = \frac{|\widetilde{T}_i\cap T_0|}{|\widetilde{T}_i|}$. Also define $\widetilde{T} = \bigcup_{i=1}^N\widetilde{T}_i \subset \{1, \ldots, n\}$. Then the general weighted $\ell_1$-minimization is formulated as

\begin{equation} \label{eqn::WL1 Nsupp}
\min_{\tilde{x}} \|\tilde{x}\|_{1,\overrightarrow{\w}} \mbox{ subject to } \|A\tilde{x}-y\|_2 \leq \epsilon, \mbox{ where } \w_k = 
  \begin{cases} 
      w_i\in[0,1] & k\in\widetilde{T}_i \\
      1 & k\in\widetilde{T}^c
   \end{cases}
   .
\end{equation}
Our main result provides recovery guarantees for (\ref{eqn::WL1 Nsupp}).  As we will discuss below, Theorem \ref{thm::WL1 Nsupp} recovers the classical un-weighted and weighted results for the single weight case. 
More importantly, in the arbitrary weight case, we show that the RIP requirements \edit{stated} here are strictly weaker than those in the classical settings when sufficiently accurate prior information is available. 

\edit{
\begin{remark}
We model the prior information with $N$ disjoint support estimates $\widetilde{T}_i$ so that, for each index (when $N=n$) or each set of indices (when $N<n$), we can apply different weights corresponding to our level of confidence that they are in the support. This framework accommodates prior information including, but not limited to, a support estimate in the traditional sense or a prior signal distribution. For example, in the event that multiple non-disjoint support estimates are available, one would simply take their intersections and set-differences to define disjoint sets and assign appropriate relative size ($\rho$) and accuracy ($\alpha$) values for these new sets. 
\end{remark}}

\begin{thm}\label{thm::WL1 Nsupp}
Let $x\in\R^n$, let $x_s$ denote its best $s$-sparse approximation, and denote the support of $x_s$ by $T_0$. Let $\widetilde{T}_i\subset \{1,\dots,n\}$ for $i=1,\dots,N$, where $1\leq N \leq n$, be arbitrary 
disjoint sets
and denote $\widetilde{T} = \bigcup_{i=1}^N \widetilde{T}_i$.  Without loss of generality, assume that the weights in \eqref{eqn::WL1 Nsupp} are ordered so that $1\geq w_1 \geq w_2 \geq \dots \geq w_N \geq 0$.  For each $i$, define the relative size $\rho_i$ and accuracy $\alpha_i$ via $|\widetilde{T}_i| = \rho_is$ and $\frac{|\widetilde{T}_i\cap T_0|}{|\widetilde{T}_i|} = \alpha_i$. Suppose that there exists $a>1$, $a\in\frac{1}{s}\mathbb{Z}$ with $\sum_{i=1}^N\rho_i(1-\alpha_i)\leq a$, and that the measurement matrix $A$ has the RIP with 
\begin{align} \label{WL1 Nsupp RIP condition}
\delta_{as} + \frac{a}{K_N^2}\delta_{(a+1)s} < \frac{a}{K_N^2} - 1,
\end{align}
where 
\begin{align}\label{K_N}
K_N &= K_N(w_1,\dots,w_N,\rho_1,\dots,\rho_N,\alpha_1,\dots,\alpha_N) \notag\\
&:= w_N + (1-w_1)\sqrt{1 + \sum_{i=1}^N (\rho_i - 2\alpha_i\rho_i)} + \sum_{j=2}^N \left((w_{j-1}-w_j)\sqrt{1 + \sum_{i=j}^N(\rho_i-2\alpha_i\rho_i)}\right).
\end{align}
Then the minimizer $\hat{x}$ to (\ref{eqn::WL1 Nsupp}) obeys
\begin{align}\label{KN bound}
\|\hat{x} - x \|_2 \leq C_0'\epsilon + C_1's^{-1/2}\left(\|x-x_s \|_1\sum_{i=1}^N w_i + (1-\sum_{i=1}^N w_i)\|x_{\widetilde{T}^c\cap T_0^c} \|_1 - \sum_{i=1}^N \sum_{j=1,j\neq i}^N w_j\|x_{\widetilde{T}_i\cap T_0^c} \|_1 \right)
\end{align}
where $C_0'$ and $C_1'$ are well-behaved constants that depend on the measurement matrix $A$, the weights $w_i$, and the parameters $\rho_i$ and $\alpha_i$ for $i=1,\dots,N$.
\end{thm}
\begin{remark}
The constants $C_0'$ and $C_1'$ are explicitly given by
\begin{align}\label{eqn::WL1 Nsupp C}
C_0' = \frac{2(1+\frac{K_N}{\sqrt{a}})}{\sqrt{1-\delta_{(a+1)s}} - \frac{K_N}{\sqrt{a}}\sqrt{1+\delta_{as}}}, \,\,\,\,
C_1' = \frac{2a^{-1/2}(\sqrt{1-\delta_{(a+1)s}}+\sqrt{1+\delta_{as}})}{\sqrt{1-\delta_{(a+1)s}} - \frac{K_N}{\sqrt{a}}\sqrt{1+\delta_{as}}}.
\end{align}
Note that $C_0'$ and $C_1'$ are identical to $C_0$ and $C_1$ from Theorem \ref{thm::WL1}, respectively, except that $b$ is replaced by $K_N$. Therefore, $C_0'\leq C_0$ and $C_1'\leq C_1$ whenever $K_N \leq b$.
\end{remark}
\begin{remark}
In Theorem 3.3 of \cite{MansourY_MultiWeightedL1} the authors define instead of our constant $K_N$ the quantity
\begin{align} \label{gamma}
\gamma = \sum_{i=1}^N w_i - (N-1) + \sum_{i=1}^N (1-w_i)\sqrt{1+\rho_i-2\alpha_i\rho_i}.
\end{align}
Thus when $N=1$ or when $w_i=1$ for all $i$, the result of \cite{MansourY_MultiWeightedL1} indeed reduces to that in \cite{Saab_weightedL1RIP} and \cite{CandeRT_Stable}, respectively. However, consider the simple case when $N=2$, $\alpha_1=\alpha_2 = 1$, and $w_1 = w_2 = w$. Then we would expect $\gamma$ to reduce to $b$ for any $\rho_1$ and $\rho_2$ such that $\rho_1+\rho_2 = \rho$ and $\alpha=1$. However, in this setting, $\gamma$ only reduces to $b$ when $\rho_1\rho_2 = 0$. Thus, the single weight result of \cite{Saab_weightedL1RIP} is not recovered as expected. This sub-optimal behavior is further illustrated in Figure \ref{fig::compareRIPconstants} below. 
\end{remark}
\begin{remark}
If $w_i =1$ for all $i$, then $K_N = 1$ and the result reduces to the $\ell_1$-minimization result without weights (see Theorem 2 in \cite{Saab_weightedL1RIP}, which is from \cite{CandeRT_Stable}). If $w_i = w$ and $\alpha_i=\alpha$ for all $i$, and $\sum_{i=1}^N\rho_i = \rho$, then the single weight result of \cite{Saab_weightedL1RIP} is recovered.  Theorem \ref{thm::WL1 Nsupp} thus recovers classical results in the un-weighted and single weight cases.
\end{remark}
\begin{remark}
To build intuition about the term $K_N$, we can consider the expression \eqref{K_N} for small $N$.
When $N=2$, we obtain,
\begin{align*}
K_2 &= w_2 + (1-w_1)\sqrt{1+\rho_1-2\alpha_1\rho_1+\rho_2-2\alpha_2\rho_2} + (w_1-w_2)\sqrt{1+\rho_2-2\alpha_2\rho_2}.
\end{align*}
When $N=3$, we obtain,
\begin{align*}
K_3 &= w_3 + (1-w_1)\sqrt{1+\rho_1-2\alpha_1\rho_1+\rho_2-2\alpha_2\rho_2+\rho_3-2\alpha_3\rho_3} \notag\\
&\quad\quad + (w_1-w_2)\sqrt{1+\rho_2-2\alpha_2\rho_2+\rho_3-2\alpha_3\rho_3} + (w_2-w_3)\sqrt{1+\rho_3-2\alpha_3\rho_3}.
\end{align*}
If either $w_1=w_2$ or $w_2=w_3$, then $K_3$ reduces to $K_2$, as desired.

We will see that small values of $K_N$ relax the requirement on the RIP constant. Suppose that $N=2$ and $\alpha_1 = \alpha_2 = 1$. Then
$$ K_2 = w_2 + (1-w_1)\sqrt{1-\rho_1-\rho_2} + (w_1-w_2)\sqrt{1-\rho_2}. $$
Suppose $1 \geq \rho_1 \gg \rho_2 \geq 0$. Then we would want to choose $w_1$ as small as possible (or as close to $w_2$ as possible) in order to minimize the dominating term $(w_1-w_2)\sqrt{1-\rho_2}$.  Similarly, if $1 \geq \rho_2 \gg \rho_1 \geq 0$, then the dominating term is $w_2$. In order to minimize $K_2$, we would want to choose $w_2$ as small as possible. Thus we see that when the support estimates are accurate, larger values of $\rho_1$ or $\rho_2$ encourage smaller corresponding weights. We also see that smaller values of $\alpha_1$ and $\alpha_2$ would encourage larger weights (as well as $w_1\approx w_2$). This also agrees with our intuition, in that we would want to select larger weights on inaccurate support estimates.
\end{remark}
\begin{remark}
Since $\delta_{as} < \delta_{(a+1)s}$, a sufficient condition for (\ref{WL1 Nsupp RIP condition}) to hold is 
\begin{align}\label{RIP condition N}
\delta_{(a+1)s} <  \frac{a-K_N^2}{a+K_N^2} := \delta^{K_N}.
\end{align}
Note that this is the same sufficient condition as seen in (\ref{WL1 RIP condition sufficient}), except with $b$ replaced by $K_N$.
\end{remark}
\edit{
\begin{remark}
\editt{If our goal is to weaken the restriction on the RIP constant and we assume} $\alpha_i$ and $\rho_i$ are known for each $i$, one could choose the weights to minimize the non-negative quantity $K_N$ and hence optimize the sufficient RIP condition (\ref{RIP condition N}). Minimizing $K_N$ \editt{in (\ref{K_N}) subject to the constraint that $0\leq w_i \leq 1$ for each $i$} is a linear program, which can be solved using standard techniques,  and it is well known that the solution will occur at a corner of the feasible region. For us, this means each of the optimal weights $w_i$ will take on the binary values 0 or 1. As an example, suppose $\rho_1=\rho_2 = 0.5$, $\alpha_1 = 0.1$, and $\alpha_2 = 0.9$. Then, solving the described optimization gives $w_1 = 1$, $w_2 = 0$, and $K_N = 0.78$. Of course, a drawback of this approach to selecting the weights is that it relies on knowledge of the $\rho_i$ and $\alpha_i$ parameters and does not necessarily imply the recovery is optimal. Moreover, the choice of weights will also impact the error bound (\ref{KN bound}). While the determination of the weights to obtain optimal recovery is an interesting  question, some heuristic options for selecting them are presented with our numerical experiments in Section \ref{sec::weighted experiments}.
\end{remark}}
To formalize the above remarks, we consider the simple case when all accuracies $\alpha_i$ are the same value, and show that as long as the accuracy level is high enough (greater than 1/2 to be precise) and the weights are smaller than that used in the single weight case, that the RIP requirements of Theorem \ref{thm::WL1 Nsupp} are strictly weaker than previous results for the single weight case.
The following Proposition shows that the smallest weight is most beneficial in relaxing the sufficient RIP condition, the largest weight is least beneficial, and a combination of weights in between produces an intermediate RIP condition.  This matches intuition, since if the support estimate is accurate\edit{,} one of course should use aggressive (small) weights on that set to encourage non-zero entries in the recovery.  On the other hand, if one is only confident about portions of the support, this proposition shows that by using non-uniform weights, one can do much better than simply selecting a single conservative weight. 
\begin{prop}
Define $\delta^b$ and $\delta^{K_N}$ as in (\ref{WL1 RIP condition sufficient}) and (\ref{RIP condition N}), respectively. Let $w_1\geq w_2\geq \dots \geq w_N$, $\sum_{i=1}^N \rho_i = \rho$ and $\alpha_1 = \alpha_2 = \dots = \alpha_N = \alpha$. For fixed $a$, $\alpha$, $\rho$, and $\rho_i$ for $i=1,\dots,N$, $\delta^b = \delta^b(w)$ and $\delta^{K_N} = \delta^{K_N}(w_1,\dots,w_N)$. Then $\delta^b(w_1) \leq \delta^{K_N}(w_1,\dots,w_N) \leq \delta^b(w_N)$ if and only if $\alpha\geq \frac{1}{2}$.
\end{prop}
\begin{proof}
Define all terms as stated in the Proposition. Then, for fixed $a$, $\alpha$, $\rho$, and $\rho_i$ for $i=1,\dots,N$, 
$$b(w,\rho,\alpha) = b(w) \quad \mbox{ and } \quad K_N(w_1,\dots,w_N,\rho_1,\dots,\rho_N,\alpha_i,\dots,\alpha_N) = K_N(w_1,\dots,w_N).$$
It is sufficient to show $b(w_N) \leq K_N(w_1,\dots,w_N) \leq b(w_1)$ if and only if $\alpha\geq \frac{1}{2}$. Since $w_1\geq w_N$, it is quickly seen that $b(w_N) \leq b(w_1)$ if and only if $\alpha\geq \frac{1}{2}$.
Next, observe that $K_N(w_1,\dots,w_N) \leq b(w_1)$ holds if and only if
\begin{align*}
&w_N + (1-w_1)\sqrt{1 + \rho-2\alpha\rho} + \sum_{j=2}^N \left((w_{j-1}-w_j)\sqrt{1 + \sum_{i=j}^N(\rho_i-2\alpha\rho_i)}\right) \leq w_1 + (1-w_1)\sqrt{1+\rho-2\alpha\rho}
\end{align*}
which is equivalent to 
\begin{align*}
\sum_{j=2}^N \left((w_{j-1}-w_j)\sqrt{1 + \sum_{i=j}^N(\rho_i-2\alpha\rho_i)}\right) \leq (w_1-w_2) + (w_2-w_3) + \dots + (w_{N-1}-w_N) 
\end{align*}
and to
\begin{align*}
(w_1-w_2)\left(\sqrt{1 + \sum_{i=2}^N(\rho_i-2\alpha\rho_i)}-1 \right) + \dots + (w_{N-1}-w_N)\left(\sqrt{1 + \sum_{i=N}^N(\rho_i-2\alpha\rho_i)}-1\right) \leq 0.
\end{align*}
For the above inequality to hold, we must have $\sqrt{1 + \sum_{i=j}^N(\rho_i-2\alpha\rho_i)}-1\leq 0$ for some $j\in\{2,\dots,N\}$. However, this requirement for any $j$ means $\sum_{i=j}^N(\rho_i-2\alpha\rho_i) \leq 0$. For \textit{any} terms in this sum to be zero or less, we must have that $\alpha \geq \frac{1}{2}$.
Similarly, the inequality $b(w_N)\leq K_N(w_1,\dots,w_N) $ holds if and only if
\begin{align*}
&w_N + (1-w_N)\sqrt{1+\rho-2\alpha\rho} \leq w_N + (1-w_1)\sqrt{1 + \rho-2\alpha\rho} + \sum_{j=2}^N \left((w_{j-1}-w_j)\sqrt{1 + \sum_{i=j}^N(\rho_i-2\alpha\rho_i)}\right)
\end{align*}
which is equivalent to 
\begin{align*}
(w_1-w_2+w_2-w_3+\dots+w_{N-1}-w_N)\sqrt{1+\rho-2\alpha\rho} \leq \sum_{j=2}^N \left((w_{j-1}-w_j)\sqrt{1 + \sum_{i=j}^N(\rho_i-2\alpha\rho_i)}\right)
\end{align*}
and to
\begin{align*}
0 \leq (w_1-w_2)&\left(\sqrt{1 + \sum_{i=2}^N(\rho_i-2\alpha\rho_i)}-\sqrt{1+\rho-2\alpha\rho} \right) + \dots \\
&+ (w_{N-1}-w_N)\left(\sqrt{1 + \rho_N-2\alpha\rho_N}-\sqrt{1+\rho-2\alpha\rho} \right). \\
\end{align*}
For the above inequality to hold, we must have $\sqrt{1 + \sum_{i=j}^N(\rho_i-2\alpha\rho_i)}-\sqrt{1+\rho-2\alpha\rho} \geq 0$ for some $j\in\{2,\dots,N\}$. However, this requirement for any $j$ means $1 + \sum_{i=j}^N(\rho_i-2\alpha\rho_i) \geq 1+\rho-2\alpha\rho$ and thus  $(\sum_{i=j}^N \rho_i - \rho)(1-2\alpha)\geq 0$.  Since $\sum_{i=j}^N \rho_i \leq \rho$ for all $j=2,\dots,N$,  we must have $\alpha \geq \frac{1}{2}$. Note that these results are tight since $K_N(w_N,\dots,w_N) = b(w_N)$ and $K_N(w_1,\dots,w_1) = b(w_1)$. 
\end{proof}

Figure \ref{fig::compareRIPconstants} compares the value of $\delta^b$ defined in (\ref{WL1 RIP condition sufficient}) when a single weight is used versus $\delta^{K_N}$ defined in (\ref{RIP condition N}) when two or three distinct weights are used as a function of the support estimate sizes. We set $\alpha_i = \alpha = 1$ for $i=1,\dots,N$ and $\sum_{i=1}^N \rho_i = \rho = 1$, where $N=2$ for the plot on the left and $N=3$ for the plot on the right. When $N=2$, the horizontal lines indicate $\delta^b$ 
when the weight $w=0.5$ or $w=0.25$ is used on the entire support estimate; in between, we see the transition of $\delta^{K_2}$ as $\rho_1$ varies with $w_1=0.5$ and $w_2=0.25$. Note that although the horizontal axis only shows the value of $\rho_1$, this determines $\rho_2$ since we take $\rho_1+\rho_2 = 1$. For comparison, we also include the value of 
\begin{align}\label{gamma sufficient condition}
\delta^\gamma := \frac{a-\gamma^2}{a+\gamma^2}
\end{align}
from \cite{MansourY_MultiWeightedL1}, where $\gamma$ is defined in (\ref{gamma}). Indeed, $\delta^\gamma$ behaves as expected at the endpoints $\rho_1=0$ and $\rho_1=1$, but falls below $\delta^b$ with $w=0.5$ for many values of $\rho_1\in (0,1)$. This again highlights the sub-optimality of the prior result \cite{MansourY_MultiWeightedL1} and the improvement offered by Theorem \ref{thm::WL1 Nsupp}.
When $N=3$, we again see the transition in $\delta^{K_3}$ 
as $\rho_1$ and $\rho_2$ are varied (which again determines $\rho_3$) with $w_1=0.5$, $w_2=0.4$, and $w_3=0.25$. The value of $\delta^b$ when a single weight is used is included for comparison. In both cases, we see that the smallest weight results in the best (largest) RIP condition and the largest weight results in the worst (smallest) RIP condition since the accuracy is assumed to be perfect, and intermediate behavior is seen in between.  This illustrates that our main result recovers the classical results in the single weight case, and that the case of multiple weights interpolates as expected. 

\begin{figure}[!htbp]
\centering
\begin{tabular}{cc}
\includegraphics[height=2.25in,width=2.7in]{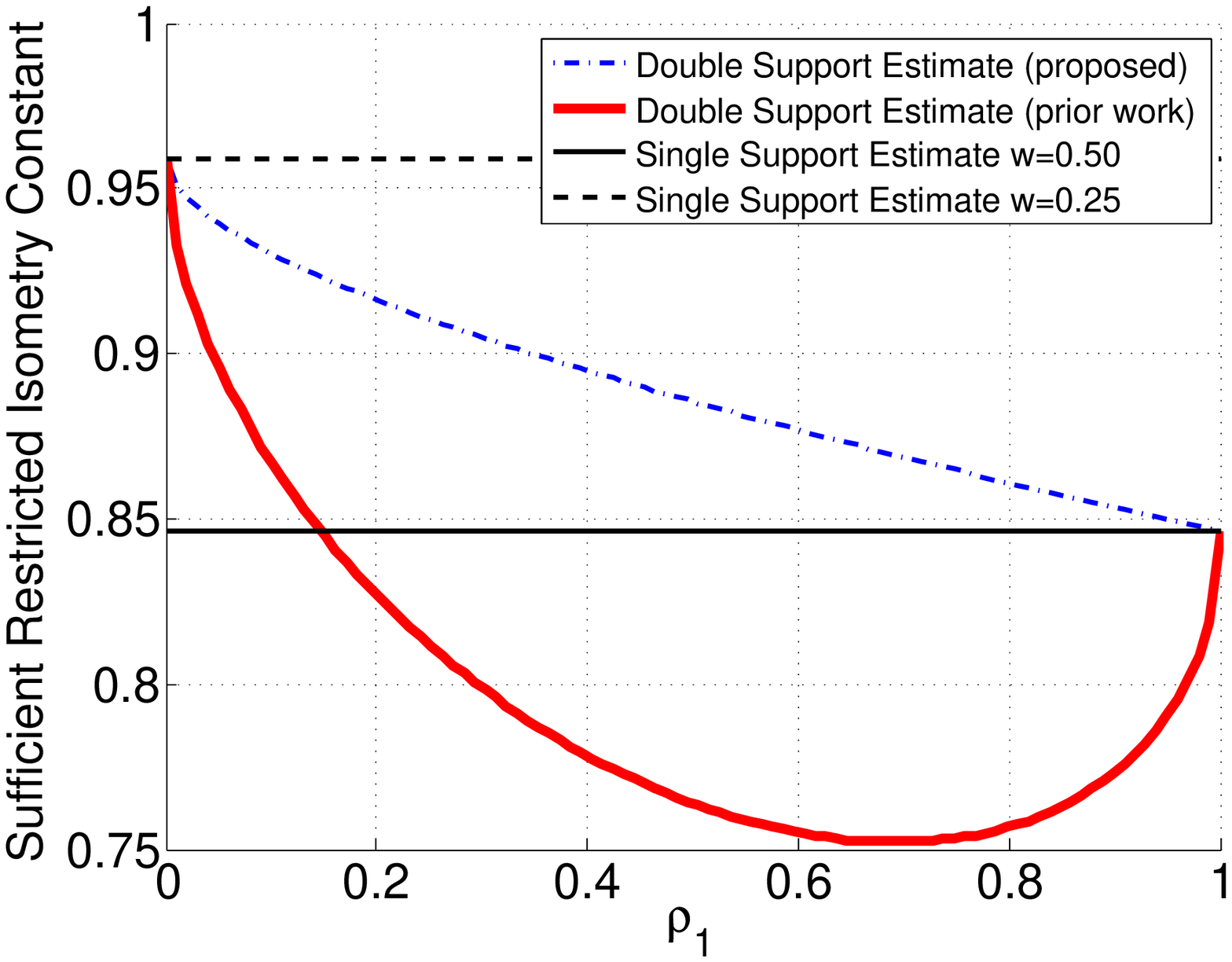} &
\includegraphics[height=2.25in,width=2.9in]{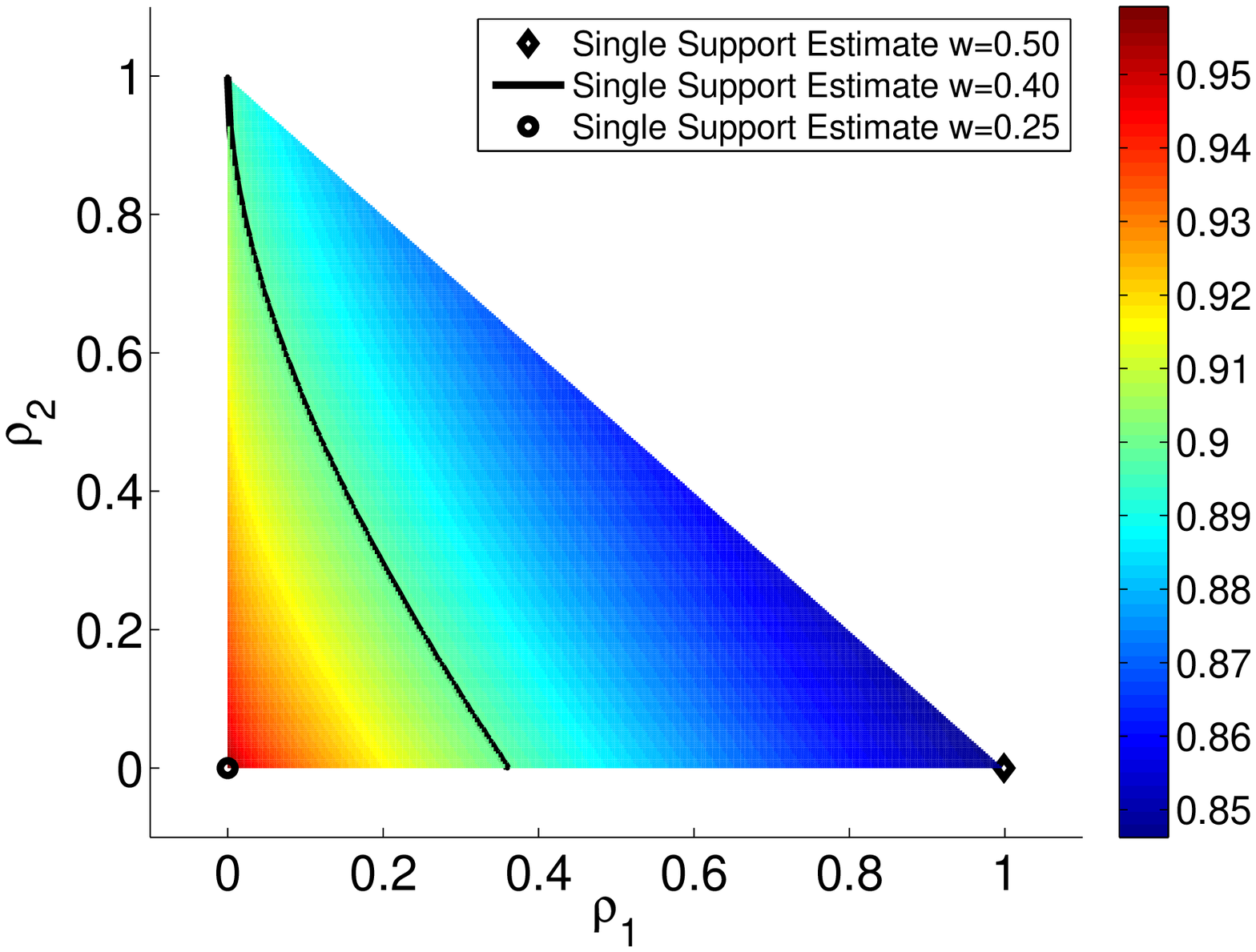} \\
(a) $N=2$ & (b) $N=3$
\end{tabular}
\caption{Comparison of $\delta^b$ (defined in (\ref{WL1 RIP condition sufficient})) and $\delta^{K_N}$ (defined in (\ref{RIP condition N})) with $a=3$. (a) We set $\alpha_1=\alpha_2=\alpha=1$, $\rho_1+\rho_2 = \rho = 1$, $w_1 = 0.5$, and $w_2=0.25$. The value of $\delta^{K_2}$ (dash-dotted line) 
is compared to $\delta^b$ in the single weight case with $w=0.5$ (solid horizontal line) and $w=0.25$ (dashed horizontal line). The value of $\delta^\gamma$ (from \cite{MansourY_MultiWeightedL1} and defined in (\ref{gamma sufficient condition})) is included for comparison (thick solid line). (b) We set $\alpha_1=\alpha_2=\alpha_3=\alpha=1$, $\rho_1+\rho_2 + \rho_3 = \rho = 1$, $w_1 = 0.5$, $w_2=0.4$, and $w_3=0.25$. The value of $\delta^{K_3}$ (indicated by the color) is compared to $\delta^b$ in the single weight case with $w=0.5$ (diamond marker), $w=0.4$ (black line), and $w=0.25$ (circle marker).}
\label{fig::compareRIPconstants}
\end{figure}

\subsection{Proof of Main Result} \label{sec::proof}

We now present the proof of Theorem \ref{thm::WL1 Nsupp}, which is inspired by that in \cite{Saab_weightedL1RIP} and \cite{CandeRT_Stable}. Let $\hat{x} = x+h$ be the minimizer of (\ref{eqn::WL1 Nsupp}), and let $T_0$ denote the set of the largest $s$ coefficients of $x$ in magnitude. Our goal is to bound the norm of the error $h$. 

{\bfseries Proof Roadmap.}  We will proceed with a sequence of lemmas, and then combine the results of the lemmas to obtain the final error bound. Briefly, the proof is organized as follows:

\begin{itemize}
\item \textit{Lemma \ref{lemma::cone constraint} - Cone Constraint}: The main challenge of the proof lies here, where we provide a cone constraint on $\|h_{T_0^c} \|_1$. This constraint is critical because it ultimately determines the parameter $K_N$ given in (\ref{K_N}). 

\item  \textit{Lemma \ref{lemma::bounding the tail} - Bounding the Tail}: Sorting the coefficients of $h_{T_0^c}$ together with Lemma \ref{lemma::cone constraint} allows us to bound the tail of $h$.

\item \textit{Lemma \ref{lemma::consequence of the rip} - Consequence of the RIP}:  Due to the RIP assumption on $A$, along with the previous lemmas, we are able to define $K_N$ and bound the largest portion of $h$.

\item Combining Lemma \ref{lemma::bounding the tail} and Lemma \ref{lemma::consequence of the rip}, we obtain the final bound on the error $\|h\|_2$.

\end{itemize}

{\bfseries Proof Notation.}
We instate the notation of Theorem \ref{thm::WL1 Nsupp}; let $|\widetilde{T}_i| = \rho_is$ for $i=1,\dots,N$, where $0\leq\rho_i\leq a_i$ and $\sum_{i=1}^N a_i > 1$. Define the accuracy of the support estimates $\alpha_i = \frac{|\widetilde{T}_i\cap T_0|}{|\widetilde{T}_i|}$ for $i=1,\dots,N$. Set $\widetilde{T} = \bigcup_{i=1}^N \widetilde{T}_i$. For ease of notation, let us also define
$$\W = \sum_{i=1}^N{w_i}$$
and 
$$ D =  \W\|x_{T_0^c} \|_1 + \left(1-\W\right)\|x_{\widetilde{T}^c\cap T_0^c} \|_1 - \sum_{i=1}^N (\W-w_i)\|x_{\widetilde{T}_i\cap T_0^c} \|_1. $$
Proceeding as in \cite{Saab_weightedL1RIP}, we will sort the coefficients of $h_{T_0^c}$ by partitioning $T_0^c$ into disjoint sets $T_j$, $j\in \{1,2,\dots\}$ each of size $as$, where $a \geq \sum_{i=1}^N a_i >1$ and $a\in \frac{1}{s}\mathbb{Z}$ (to ensure the cardinality of each $T_j$ is an integer). That is, $T_1$ indexes the $as$ largest in magnitude coefficients of $h_{T_0^c}$, $T_2$ indexes the second $as$ largest in magnitude coefficients of $h_{T_0^c}$, and so on. Define $T_{01} = T_0 \cup T_1$. 

$\;$\\
\indent We now prove each of the above mentioned lemmas in sequence, and their combination will complete the proof.

\begin{lemma}[Cone Constraint]\label{lemma::cone constraint}
The vector $h$ obeys the following cone constraint,
\begin{align} \label{option4}
\|h_{T_0^c} \|_1 &\leq  w_N\|h_{T_0} \|_1 + (1-w_1)\|h_{T_0\cup \bigcup_{i=1}^N \widetilde{T}_i \setminus \bigcup_{i=1}^N (\widetilde{T}_i\cap T_0)}\|_1 + \sum_{j=2}^N(w_{j-1}-w_j)\|h_{T_0\cup \bigcup_{i=j}^N \widetilde{T}_i \setminus \bigcup_{i=j}^N (\widetilde{T}_i\cap T_0)}\|_1 +2D.
\end{align}
\end{lemma}
\begin{proof}
Since $\hat{x}=x+h$ is a minimizer of (\ref{eqn::WL1 Nsupp}), then $ \|x+h\|_{1,\overrightarrow{\w}} \leq \|x \|_{1,\overrightarrow{\w}}$. By the choice of weights, we have
$$ \sum_{i=1}^N w_i \|x_{\widetilde{T}_i}+h_{\widetilde{T}_i} \|_1 + \|x_{\widetilde{T}^c}+h_{\widetilde{T}^c} \|_1 \leq \sum_{i=1}^N w_i \|x_{\widetilde{T}_i} \|_1 + \|x_{\widetilde{T}^c} \|_1. $$
Furthermore, we have
\begin{align*}
\sum_{i=1}^N (w_i \|x_{\widetilde{T}_i\cap T_0}+&h_{\widetilde{T}_i\cap T_0} \|_1 + w_i \|x_{\widetilde{T}_i\cap T_0^c}+h_{\widetilde{T}_i\cap T_0^c} \|_1)
+ \|x_{\widetilde{T}^c\cap T_0}+h_{\widetilde{T}^c\cap T_0} \|_1 + \|x_{\widetilde{T}^c\cap T_0^c}+h_{\widetilde{T}^c\cap T_0^c} \|_1\\
&\leq \sum_{i=1}^N (w_i \|x_{\widetilde{T}_i\cap T_0} \|_1 + w_i \|x_{\widetilde{T}_i\cap T_0^c} \|_1)
+ \|x_{\widetilde{T}^c\cap T_0} \|_1 + \|x_{\widetilde{T}^c\cap T_0^c} \|_1.
\end{align*}
Next, we use the reverse triangle inequality to get
\begin{align}\label{eqn::reverse triangle}
\sum_{i=1}^N w_i\|h_{\widetilde{T}_i\cap T_0^c} \|_1 + \|h_{\widetilde{T}^c\cap T_0^c} \|_1
\leq \sum_{i=1}^N w_i \|h_{\widetilde{T}_i\cap T_0} \|_1 + \|h_{\widetilde{T}^c\cap T_0} \|_1 + 2(\sum_{i=1}^N w_i \|x_{\widetilde{T}_i\cap T_0^c} \|_1 + \|x_{\widetilde{T}^c\cap T_0^c} \|_1).
\end{align}

Now, we can write $\|h_{T_0^c}\|_1 = \sum_{i=1}^N \|h_{\widetilde{T}_i\cap T_0^c} \|_1 + \|h_{\widetilde{T}^c\cap T_0^c} \|_1$. 
Let us add and subtract $w_i\|h_{\widetilde{T}_j\cap T_0^c} \|_1$ for all pairs of $i$ and $j$ such that $i,j=1,\dots,N$ and $i\neq j$, and $w_i\|h_{\widetilde{T}^c \cap T_0^c} \|_1$ for $i=1,\dots,N$ to the left side of (\ref{eqn::reverse triangle}). Then the left side of (\ref{eqn::reverse triangle}) becomes
$$ \W \|h_{T_0^c} \|_1 + (1-\W) \|h_{\widetilde{T}^c\cap T_0^c} \|_1 - \sum_{i=1}^N (\W-w_i)\|h_{\widetilde{T}_i\cap T_0^c} \|_1.$$
Similarly, we can write $\|h_{T_0}\|_1 = \sum_{i=1}^N \|h_{\widetilde{T}_i\cap T_0} \|_1 + \|h_{\widetilde{T}^c\cap T_0} \|_1$. 
Let us add and subtract $w_i\|h_{\widetilde{T}_j\cap T_0} \|_1$ for all pairs of $i$ and $j$ such that $i,j=1,\dots,N$ and $i\neq j$, and $w_i\|h_{\widetilde{T}^c \cap T_0} \|_1$ for $i=1,\dots,N$ to the right side of (\ref{eqn::reverse triangle}), as well as $w_i\|x_{\widetilde{T}_j\cap T_0^c} \|_1$ for all pairs of $i$ and $j$ such that $i,j=1,\dots,N$ and $i\neq j$, and $w_i\|x_{\widetilde{T}^c \cap T_0^c} \|_1$ for $i=1,\dots,N$. Then the right side of (\ref{eqn::reverse triangle}) becomes 
\begin{align*}
\W &\|h_{T_0} \|_1 + (1-\W) \|h_{\widetilde{T}^c\cap T_0} \|_1 - \sum_{i=1}^N (\W-w_i)\|h_{\widetilde{T}_i\cap T_0} \|_1 \\
&+ 2(\W \|x_{T_0^c} \|_1 + (1-\W) \|x_{\widetilde{T}^c\cap T_0^c} \|_1 - \sum_{i=1}^N (\W-w_i)\|x_{\widetilde{T}_j\cap T_0^c} \|_1).
\end{align*}
Putting these together, and using our definition of $D$, we have
\begin{align}\label{20}
\W \|h_{T_0^c} \|_1 &+ (1-\W) \|h_{\widetilde{T}^c\cap T_0^c} \|_1 - \sum_{i=1}^N (\W-w_i)\|h_{\widetilde{T}_i\cap T_0^c} \|_1 \notag \\
&\leq \W \|h_{T_0} \|_1 + (1-\W) \|h_{\widetilde{T}^c\cap T_0} \|_1 - \sum_{i=1}^N (\W-w_i)\|h_{\widetilde{T}_i\cap T_0} \|_1 +2D.
\end{align}

But, we can also write $\|h_{T_0^c} \|_1$ as 
$$ \|h_{T_0^c} \|_1 = \W \|h_{T_0^c} \|_1 + \sum_{i=1}^N \left((1-\W) \|h_{\widetilde{T}_i\cap T_0^c}\|_1\right) + (1-\W)\|h_{\widetilde{T}^c\cap T_0^c}\|_1.$$
Solving for $\W \|h_{T_0^c} \|_1$ and substituting into (\ref{20}) gives
\begin{align*}
\|h_{T_0^c} \|_1 &-  \sum_{i=1}^N \left((1-\W) \|h_{\widetilde{T}_i\cap T_0^c}\|_1\right) - (1-\W)\|h_{\widetilde{T}^c\cap T_0^c}\|_1
+ (1-\W) \|h_{\widetilde{T}^c\cap T_0^c} \|_1 - \sum_{i=1}^N (\W-w_i)\|h_{\widetilde{T}_i\cap T_0^c} \|_1 \\
&\leq \W \|h_{T_0} \|_1 + (1-\W) \|h_{\widetilde{T}^c\cap T_0} \|_1 - \sum_{i=1}^N (\W-w_i)\|h_{\widetilde{T}_i\cap T_0} \|_1 +2D.
\end{align*}
Simplifying, we get 
\begin{align} \label{21}
\|h_{T_0^c} \|_1&\leq \sum_{i=1}^N \left((1-\W) \|h_{\widetilde{T}_i\cap T_0^c}\|_1\right) + \sum_{i=1}^N (\W-w_i)\|h_{\widetilde{T}_i\cap T_0^c} \|_1 \notag\\
&\quad\quad+ \W \|h_{T_0} \|_1 + (1-\W) \|h_{\widetilde{T}^c\cap T_0} \|_1 - \sum_{i=1}^N (\W-w_i)\|h_{\widetilde{T}_i\cap T_0} \|_1 +2D \\
&= \sum_{i=1}^N(1-w_i)\|h_{\widetilde{T}_i\cap T_0^c} \|_1 +  \W \|h_{T_0} \|_1 + \|h_{\widetilde{T}^c\cap T_0} \|_1 \notag \\
&\quad\quad  - \sum_{i=1}^N w_i\left(\|h_{\widetilde{T}^c\cap T_0} \|_1 + \sum_{j=1, j\neq i}^N \|h_{\widetilde{T}_j\cap T_0} \|_1 \right) + 2D \notag \\
&= \sum_{i=1}^N(1-w_i)\|h_{\widetilde{T}_i\cap T_0^c} \|_1 +  \W \|h_{T_0} \|_1 + \|h_{\widetilde{T}^c\cap T_0} \|_1 - \sum_{i=1}^N w_i\|h_{\widetilde{T}_i^c\cap T_0} \|_1 \notag \\
&\quad\quad +\sum_{i=1}^N\|h_{\widetilde{T}_i^c\cap T_0} \|_1 - \sum_{i=1}^N\|h_{\widetilde{T}_i^c\cap T_0} \|_1 +2D \label{lemma step 1}\\
&= \W \|h_{T_0} \|_1 + \|h_{\widetilde{T}^c\cap T_0} \|_1 - \sum_{i=1}^N\|h_{\widetilde{T}_i^c\cap T_0} \|_1+ \sum_{i=1}^N(1-w_i)\left(\|h_{\widetilde{T}_i\cap T_0^c} \|_1+ \|h_{\widetilde{T}_i^c\cap T_0} \|_1\right) + 2D \notag \\
&= \left(\W -(N-1)\right) \|h_{T_0} \|_1 + \sum_{i=1}^N (1-w_i)\left(\|h_{\widetilde{T}_i^c\cap T_0} \|_1 + \|h_{\widetilde{T}_i\cap T_0^c} \|_1\right) + 2D \label{lemma step 2},
\end{align}
where in (\ref{lemma step 1}) we have added zero and observed that $\|h_{\widetilde{T}^c\cap T_0} \|_1 + \sum_{j=1, j\neq i}^N \|h_{\widetilde{T}_j\cap T_0} \|_1=\|h_{\widetilde{T}_i^c \cap T_0} \|_1$, and in (\ref{lemma step 2}) we have observed that $\sum_{i=1}^N\|h_{\widetilde{T}_i^c\cap T_0} \|_1 = (N-1)\|h_{T_0} \|_1 + \|h_{\widetilde{T}^c\cap T_0} \|_1$.
Then assuming, without loss of generality, $w_1 \geq w_2 \geq \dots \geq w_N$, and writing $1-w_i = 1-w_1+w_1-w_i$ for $i>1$, we have
\begin{align}
\|h_{T_0^c} \|_1 &\leq  \left(\W -(N-1)\right)\|h_{T_0} \|_1 + (1-w_1)\sum_{i=1}^N [\|h_{\widetilde{T}_i^c\cap T_0} \|_1 + \|h_{\widetilde{T}_i\cap T_0^c} \|_1] \notag\\
&\quad\quad + \sum_{i=2}^N (w_1-w_i)[\|h_{\widetilde{T}_i^c\cap T_0} \|_1 + \|h_{\widetilde{T}_i\cap T_0^c} \|_1] +2D.
\end{align} 
Next, write $w_1-w_i = w_1-w_2+w_2-w_i$ for $i>2$. Then we have
\begin{align}
\|h_{T_0^c} \|_1 &\leq  \left(\W -(N-1)\right)\|h_{T_0} \|_1 + (1-w_1)\sum_{i=1}^N [\|h_{\widetilde{T}_i^c\cap T_0} \|_1 + \|h_{\widetilde{T}_i\cap T_0^c} \|_1] \notag\\
&\quad\quad + (w_1-w_2)\sum_{i=2}^N[\|h_{\widetilde{T}_i^c\cap T_0} \|_1 + \|h_{\widetilde{T}_i\cap T_0^c} \|_1] + \sum_{i=3}^N (w_2-w_i)[\|h_{\widetilde{T}_i^c\cap T_0} \|_1 + \|h_{\widetilde{T}_i\cap T_0^c} \|_1] +2D.
\end{align} 
Continuing in this manner gives us
\begin{align}
\|h_{T_0^c} \|_1 &\leq  \left(\W-(N-1)\right)\|h_{T_0} \|_1 + (1-w_1)\sum_{i=1}^N [\|h_{\widetilde{T}_i^c\cap T_0} \|_1 + \|h_{\widetilde{T}_i\cap T_0^c} \|_1] \notag\\
&\quad\quad + \sum_{j=2}^N(w_{j-1}-w_j)\sum_{i=j}^N[\|h_{\widetilde{T}_i^c\cap T_0} \|_1 + \|h_{\widetilde{T}_i\cap T_0^c} \|_1]  + 2D.
\end{align} 
Noting that $\|h_{\widetilde{T}_i^c\cap T_0} \|_1 = \sum_{j=1,j\neq i}^N\|h_{\widetilde{T}_j\cap T_0} \|_1 + \|h_{T_0\cap \bigcap_{j=1}^N\widetilde{T}_j^c} \|_1$ and $\|h_{T_0} \|_1=\sum_{i=1}^N\|h_{\widetilde{T}_i\cap T_0} \|_1 + \|h_{T_0\cap \bigcap_{j=1}^N\widetilde{T}_j^c} \|_1$, and that we can write $\sum_{i=j}^N \|h_{\widetilde{T}_i\cap T_0^c} \|_1 + \|h_{T_0\cap \bigcap_{j=1}^N\widetilde{T}_j^c} \|_1 = \|h_{T_0\cup \bigcup_{i=j}^N \widetilde{T}_i \setminus \bigcup_{i=j}^N (\widetilde{T}_i\cap T_0)} \|_1$ for any $j=1,\dots,N$, the above can also be expressed as
\begin{align}\label{the above0}
\|h_{T_0^c} \|_1 &\leq  \left(\W -(N-1)\right)\|h_{T_0} \|_1 + (1-w_1)[(N-1)\|h_{T_0} \|_1 + \|h_{T_0\cup \bigcup_{i=1}^N \widetilde{T}_i \setminus \bigcup_{i=1}^N (\widetilde{T}_i\cap T_0)}\|_1] \notag\\
&\quad\quad + \sum_{j=2}^N(w_{j-1}-w_j)[(N-j)\|h_{T_0} \|_1 + \|h_{T_0\cup \bigcup_{i=j}^N \widetilde{T}_i \setminus \bigcup_{i=j}^N (\widetilde{T}_i\cap T_0)}\|_1]  + 2D.
\end{align} 
Combining all coefficients of $\|h_{T_0}\|_1$, we have
\begin{align}\label{the above}
\sum_{i=1}^N w_i - (&N-1) + (1-w_1)(N-1) + \sum_{j=2}^N (w_{j-1}-w_j)(N-j) \nonumber\\
&= \sum_{i=1}^N w_i  - (N-1)w_1 + (N-2)w_1 + \sum_{j=2}^{N-1} (N-(j+1))w_j - \sum_{j=2}^{N-1} (N-j)w_j \nonumber\\
&= \sum_{i=2}^N w_i  - \sum_{j=2}^{N-1} w_j \nonumber\\
& = w_N. 
\end{align}
Finally, combining \eqref{the above0} with \eqref{the above} we  arrive at (\ref{option4}).
\end{proof}

\begin{lemma}[Bounding the Tail] \label{lemma::bounding the tail}
We have the following bound on the tail of the error $h$,
\begin{align} \label{24}
\|h_{T_{01}^c} \|_2 &\leq (as)^{-1/2} \left[w_N\|h_{T_0} \|_1 + (1-w_1)\|h_{T_0\cup \bigcup_{i=1}^N \widetilde{T}_i \setminus \bigcup_{i=1}^N (\widetilde{T}_i\cap T_0)}\|_1 \right.\notag\\
&\quad\quad \left. + \sum_{j=2}^N(w_{j-1}-w_j)\|h_{T_0\cup \bigcup_{i=j}^N \widetilde{T}_i \setminus \bigcup_{i=j}^N (\widetilde{T}_i\cap T_0)}\|_1+ 2D \right].
\end{align}
\end{lemma}
\begin{proof}
As in \cite{CandeRT_Stable}, note that
\begin{align}\label{lemma1::eqn1}
\|h_{T_j}\|_2 \leq \sqrt{as}\|h_{T_j} \|_{\infty} \leq (as)^{-1/2}\|h_{T_{j-1}} \|_1,
\end{align}
where we have observed that (by construction) the average of the terms (in magnitude) on $h_{T_{j-1}}$ must be at least as large as $\|h_{T_j} \|_{\infty}$.
Then noting that $h_{T_{0}^c} = \sum_{j\geq 1} h_{T_j}$ and  $h_{T_{01}^c} = \sum_{j\geq 2} h_{T_j}$ and using the triangle inequality along with (\ref{lemma1::eqn1}), we have
\begin{align} \label{23}
\|h_{T_{01}^c} \|_2 &\leq \sum_{j\geq 2} \| h_{T_j} \|_2  \notag\\
&\leq (as)^{-1/2}\sum_{j\geq 1} \| h_{T_j} \|_1 \notag\\
&\leq (as)^{-1/2}\|h_{T_0^c} \|_1.
\end{align}
Combining (\ref{23}) with (\ref{option4}) gives (\ref{24}).
\end{proof}

\begin{lemma}[Consequence of the RIP] \label{lemma::consequence of the rip}
Define
\begin{align}\label{def::K}
K_N &= w_N + (1-w_1)\sqrt{1 + \sum_{i=1}^N (\rho_i - 2\alpha_i\rho_i)} + \sum_{j=2}^N \left((w_{j-1}-w_j)\sqrt{1 + \sum_{i=j}^N(\rho_i-2\alpha_i\rho_i)}\right).
\end{align}
Then the following inequality holds,
\begin{align} \label{25}
\|h_{T_{01}}\|_2 \leq \frac{2\epsilon + \frac{2\sqrt{1+\delta_{as}}}{\sqrt{as}}D}{\sqrt{1-\delta_{(a+1)s}} - \frac{\sqrt{1+\delta_{as}}}{\sqrt{a}}K_N},
\end{align}
when the denominator is positive.
\end{lemma}
\begin{proof}
Again, following \cite{CandeRT_Stable}, and noting that since $A$ satisfies the RIP and $\|Ah\|_2 \leq 2\epsilon$ due to the feasibility of both $x$ and $\hat{x}$, we have
\begin{align*}
\sqrt{1-\delta_{(a+1)s}}\|h_{T_{01}} \|_2 &\leq \|Ah_{T_{01}} \|_2 \\
&= \|Ah_{T_{01}} + Ah_{T_{01}^c}- Ah_{T_{01}^c}\|_2\\
&\leq \| Ah\|_2 + \|Ah_{T_{01}^c}\|_2\\
&= 2\epsilon +  \|\sum_{j\geq 2} Ah_{T_j}\|_2\\
&\leq 2\epsilon + \sum_{j\geq 2} \|Ah_{T_j} \|_2\\
&\leq 2\epsilon + \sqrt{1+\delta_{as}}\sum_{j\geq2}\|h_{T_j} \|_2 \\
&\leq 2\epsilon + \frac{\sqrt{1+\delta_{as}}}{\sqrt{as}}\|h_{T_{0}^c} \|_1,
\end{align*}
where the last inequality follows from (\ref{23}).
Then applying (\ref{option4}), we get
\begin{align*}
&\sqrt{1-\delta_{(a+1)s}}\|h_{T_{01}} \|_2 \leq 2\epsilon + 2\frac{\sqrt{1+\delta_{as}}}{\sqrt{as}}D \\
&\quad\quad + \frac{\sqrt{1+\delta_{as}}}{\sqrt{as}}\left[w_N\|h_{T_0} \|_1 + (1-w_1)\|h_{T_0\cup \bigcup_{i=1}^N \widetilde{T}_i \setminus \bigcup_{i=1}^N (\widetilde{T}_i\cap T_0)}\|_1 + \sum_{j=2}^N(w_{j-1}-w_j)\|h_{T_0\cup \bigcup_{i=j}^N \widetilde{T}_i \setminus \bigcup_{i=j}^N (\widetilde{T}_i\cap T_0)}\|_1\right]. 
\end{align*}

Note that $|T_0\cup \bigcup_{i=j}^N \widetilde{T}_j \setminus \bigcup_{i=j}^N (\widetilde{T}_i\cap T_0)| = s + \sum_{i=j}^N (\rho_i s - 2\alpha_i\rho_i s)$ for any $j=1,\dots,N$. Thus, we have
$$ \|h_{T_0\cup \bigcup_{i=j}^N \widetilde{T}_i \setminus \bigcup_{i=j}^N (\widetilde{T}_i\cap T_0)}\|_1 \leq \sqrt{s + s\sum_{i=j}^N (\rho_i-2\alpha_i\rho_i)} \|h_{T_0\cup \bigcup_{i=j}^N \widetilde{T}_i \setminus \bigcup_{i=j}^N (\widetilde{T}_i\cap T_0)}\|_2 $$
for any $j=1,\dots,N$. 
Now, since $T_1$ contains the largest $as$ coefficients of $h_{T_0^c}$ with $a>1$, and $|\bigcup_{i=j}^N \widetilde{T}_i \setminus \bigcup_{i=j}^N (\widetilde{T}_i\cap T_0)| = \sum_{i=j}^N (\rho_i s - \alpha_i\rho_i s) = s\sum_{i=1}^N \rho_i(1-\alpha_i)$, then $\|h_{T_0\cup \bigcup_{i=j}^N \widetilde{T}_i \setminus \bigcup_{i=j}^N (\widetilde{T}_i\cap T_0)} \|_2 \leq \|h_{T_{01}} \|_2$ for any $j=1,\dots,N$ as long as we require $\sum_{i=1}^N \rho_i(1-\alpha_i) \leq a$. Note that since $\rho_i\leq a_i$ for $i=1,\dots,N$, $\sum_{i=1}^N \rho_i(1-\alpha_i)\leq \sum_{i=1}^N a_i(1-\alpha_i) \leq \sum_{i=1}^N a_i \leq a$ as long as $\alpha_i \geq 0$ for $i=1,\dots,N$. Consequently, we can bound
$$ \|h_{T_0\cup \bigcup_{i=j}^N \widetilde{T}_i \setminus \bigcup_{i=j}^N (\widetilde{T}_i\cap T_0)}\|_2 \leq \|h_{T_{01}} \|_2$$
for any $j=1,\dots,N$. Also, $\| h_{T_0} \|_1 \leq \sqrt{s}\|h_{T_0} \|_2 \leq \sqrt{s}\|h_{T_{01}} \|_2$.
Therefore, we have
\begin{align}\label{l2bound}
&\sqrt{1-\delta_{(a+1)s}}\|h_{T_{01}} \|_2 \leq  2\epsilon + 2\frac{\sqrt{1+\delta_{as}}}{\sqrt{as}}D \notag\\
&\quad\quad + \frac{\sqrt{1+\delta_{as}}}{\sqrt{a}}\|h_{T_{01}} \|_2\left[ w_N + (1-w_1)\sqrt{1 + \sum_{i=1}^N (\rho_i - 2\alpha_i\rho_i)} + \sum_{j=2}^N \left((w_{j-1}-w_j)\sqrt{1 + \sum_{i=j}^N(\rho_i-2\alpha_i\rho_i)}\right)\right].
\end{align}

Defining $K_N$ as in (\ref{def::K}) and solving (\ref{l2bound}) for $\|h_{T_{01}}\|_2$ gives (\ref{25}), where we assume the denominator is positive. 
\end{proof}

$\;$\\
\indent{\bfseries Putting it all together.}  Finally, using that $\| h\|_2 = \|h_{T_{01}}  + h_{T_{01}^c} \|_2 \leq \|h_{T_{01}} \|_2 + \|h_{T_{01}^c} \|_2$, along with (\ref{24}), (\ref{25}) and the arguments leading to (\ref{l2bound}), we get
\begin{align} \label{finalbound}
\|h\|_2 &\leq \|h_{T_{01}} \|_2 + \frac{1}{\sqrt{as}} \left[w_N\|h_{T_0} \|_1 + (1-w_1)\|h_{T_0\cup \bigcup_{i=1}^N \widetilde{T}_i \setminus \bigcup_{i=1}^N (\widetilde{T}_i\cap T_0)}\|_1 \right. \notag \\
&\quad \left. + \sum_{j=2}^N(w_{j-1}-w_j)\|h_{T_0\cup \bigcup_{i=j}^N \widetilde{T}_i \setminus \bigcup_{i=j}^N (\widetilde{T}_i\cap T_0)}\|_1+ 2D \right]\notag \\
&\leq  \|h_{T_{01}} \|_2 +  \frac{1}{\sqrt{as}} \left[\sqrt{s}K_N\|h_{T_{01}} \|_2+ 2D \right] \notag \\
&= \left(1+\frac{K_N}{\sqrt{a}}\right)\|h_{T_{01}} \|_2 + \frac{2}{\sqrt{as}}D \notag \\
&\leq  \left(1+\frac{K_N}{\sqrt{a}}\right) \left( \frac{2\epsilon + \frac{2\sqrt{1+\delta_{as}}}{\sqrt{as}}D}{\sqrt{1-\delta_{(a+1)s}} - \frac{\sqrt{1+\delta_{as}}}{\sqrt{a}}K_N} \right) + \frac{2}{\sqrt{as}}D \notag \\
&= \frac{(1+\frac{K_N}{\sqrt{a}})2\epsilon + \frac{2}{\sqrt{as}}(\sqrt{1+\delta_{as}} + \sqrt{1-\delta_{(a+1)s}})D}{\sqrt{1-\delta_{(a+1)s}} - \frac{\sqrt{1+\delta_{as}}}{\sqrt{a}}K_N},
\end{align}
with the condition that the denominator is positive. That is, we require
\begin{align*}
\sqrt{1-\delta_{(a+1)s}} - \frac{\sqrt{1+\delta_{as}}}{\sqrt{a}}K_N > 0,
\end{align*}
which is equivalent to
\begin{align*}
\delta_{as} + \frac{a}{K_N^2} \delta_{(a+1)s}< \frac{a}{K_N^2}-1.
\end{align*}

\section{Numerical Experiments}\label{sec::weighted experiments}

In this section, we include numerical experiments of weighted $\ell_1$-minimization and demonstrate that non-uniform weights can be preferable to a uniform weight. We stress that the purpose of these experiments is to illustrate the potential benefits of the multi-weight setup, and not to investigate the associated application areas. We first present experiments for synthetically generated signals, and then present an example where we use weighted $\ell_1$-minimization to recover a compressively sampled video signal. 

\subsection{Synthetic Experiments}
For our first set of experiments, the signal $x$ is synthetically generated and the recovery error when using weighted $\ell_1$-minimization with non-uniform weights is compared to that when using a single constant weight. Here, the signal is of dimension $n=256$ and sparsity $s=|T_0| = |\supp(x)| = 16$ with standard Gaussian nonzero values on $T_0$, the measurement matrix is Gaussian, the measurement noise $z$ is i.i.d. Gaussian with mean zero and standard deviation 0.01, and 500 trials are performed at each measurement level tested. We compare the \edit{relative} recovery error $\|x-\hat{x} \|_2\edit{/\|x \|_2}$ when using either a single support estimate set $\widetilde{T}$ or two disjoint support estimate sets $\widetilde{T}_1$ and $\widetilde{T}_2$, such that $\widetilde{T} = \widetilde{T}_1\cup\widetilde{T}_2$ (this implies $\rho = \rho_1 + \rho_2$). 

First, we set $\alpha=\alpha_1=\alpha_2=1$, $\rho=1$, and vary the sizes of $\rho_1$ and $\rho_2$ while maintaining that $\rho_1+\rho_2=1$. 
We set $w_1=0.5$ (applied on $\widetilde{T}_1$) and $w_2=0.25$ (applied on $\widetilde{T}_2$) in the two support estimate setting, and compare the recovery error to the single support estimate case when $w=0.5$ or $w=0.25$ (applied on $\widetilde{T}$). Figure \ref{fig::weightedL1_TwoSupportEstimates} displays the mean \edit{relative} recovery error of these results as a function of the number of measurements acquired. As expected, since all support estimates are completely accurate (that is, all elements of $\widetilde{T}$, $\widetilde{T}_1$, and $\widetilde{T}_2$ are elements of $T_0$), setting all weights to the smallest value of $w=0.25$ performs the best while setting all weights to the largest value of $w=0.5$ performs the worst; using a combination of these two weights as $\rho_1$ and $\rho_2$ are varied produces intermediate performance. This behavior reflects empirically the theoretical result of Figure \ref{fig::compareRIPconstants}. 

Next, we set $\alpha=0.5$, $\rho=1$, $\rho_1=\rho_2=0.5$, and vary the sizes of $\alpha_1$ and $\alpha_2$ while maintaining that $\rho_1\alpha_1 + \rho_2\alpha_2 = \rho\alpha$ (with our choice of parameters, this means $\alpha_1 + \alpha_2 = 1$). We again set $w_1=0.5$ and $w_2=0.25$ in the two support estimate setting, and compare the \edit{relative} recovery error to the single support estimate case when $w=0.5$ or $w=0.25$. In Figure \ref{fig::weightedL1_TwoSupportEstimates}, we see that the best recovery is achieved when $\alpha_1=0$ and $\alpha_2=1$. This is again expected since $w_2=0.25$ is applied on $\widetilde{T}_2$ which contains the \textit{correctly} identified elements in $T_0$ and $w_1=0.5$ is applied on $\widetilde{T}_1$ which contains the \textit{incorrectly} estimated elements in $T_0^c$. As $\alpha_1$ increases to 1 and $\alpha_2$ decreases to 0 the \edit{relative} recovery error increases since fewer correctly identified elements in $T_0$ receive the smaller weight $w_2$, but rather the larger weight $w=0.5$. The recovery results when using a single constant weight of either $w=0.5$ or $w=0.25$ or two weights with $\alpha_1=\alpha_2=0.5$ all seem similar, however, there is a subtlety. The recovery tends to be \textit{slightly} better when the single weight $w=0.25$ is used than when $w=0.5$ is used, and the two weight result with $\alpha_1=\alpha_2=0.5$ seems to fall in between the $w=0.25$ and $w=0.5$ curves. Turning to the theory, when $\alpha = 0.5$ the quantity $b=1$ and when $\alpha_1=\alpha_2=0.5$ the quantity $K_2=1$, in which case $C_0=C_0'$ and $C_1=C_1'$ and the RIP condition (\ref{WL1 Nsupp RIP condition}) of Theorem \ref{thm::WL1 Nsupp} is identical to (\ref{WL1 RIP condition}) in Theorem \ref{thm::WL1}. Thus, the relationship between recovery error for these settings can be explained by the terms $w\|x-x_s \|_1 + (1-w)\|x_{\widetilde{T}^c\cap T_0^c} \|_1$ and 
$$ (w_1+w_2)\|x-x_s \|_1 + (1-w_1-w_2)\|x_{\widetilde{T}^c\cap T_0^c} \|_1 - w_2\|x_{\widetilde{T}_1\cap T_0^c} \|_1 - w_1\|x_{\widetilde{T}_2\cap T_0^c} \|_1 $$
from the error bounds in Theorems \ref{thm::WL1} and \ref{thm::WL1 Nsupp}, respectively. Specifically, since $\widetilde{T} = \widetilde{T}_1\cup\widetilde{T}_2$, when $w=0.25$, $w_1=0.5$, $w_2=0.25$, and noting that $\|x-x_s \|_1 = \|x_{T_0^c} \|_1$ and $\|x_{\widetilde{T}^c\cap T_0^c} \|_1 \leq \|x_{T_0^c} \|_1$, we have
$$ 0.25\|x_{T_0^c} \|_1 + 0.75\|x_{\widetilde{T}^c\cap T_0^c} \|_1 \leq 0.75\|x_{T_0^c} \|_1 + 0.25\|x_{\widetilde{T}^c\cap T_0^c} \|_1 - 0.25\|x_{\widetilde{T}_1\cap T_0^c} \|_1 - 0.5\|x_{\widetilde{T}_2\cap T_0^c} \|_1. $$
Similarly, when $w=0.5$, we have
$$ 0.75\|x_{T_0^c} \|_1 + 0.25\|x_{\widetilde{T}^c\cap T_0^c} \|_1 - 0.25\|x_{\widetilde{T}_1\cap T_0^c} \|_1 - 0.5\|x_{\widetilde{T}_2\cap T_0^c} \|_1  \leq 0.5\|x_{T_0^c} \|_1 + 0.5\|x_{\widetilde{T}^c\cap T_0^c} \|_1. $$
Although satisfying the above relationships, these terms are quite close, which is reflected in the closeness of the curves in Figure \ref{fig::weightedL1_TwoSupportEstimates}.

\begin{figure}[!htbp]
\centering
\begin{tabular}{cc}

\includegraphics[height=1.93in,width=2.55in]{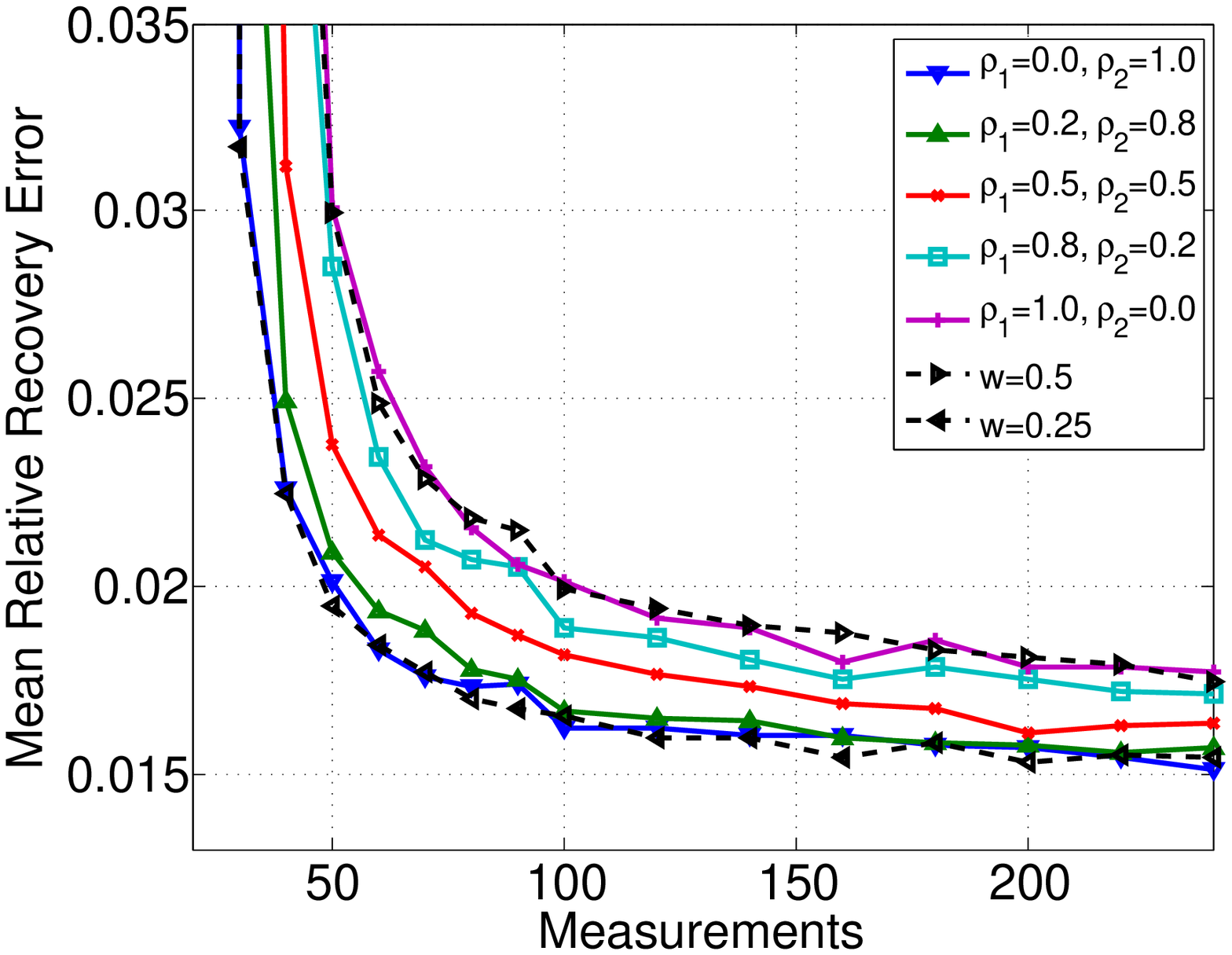} &
\includegraphics[height=1.93in,width=2.55in]{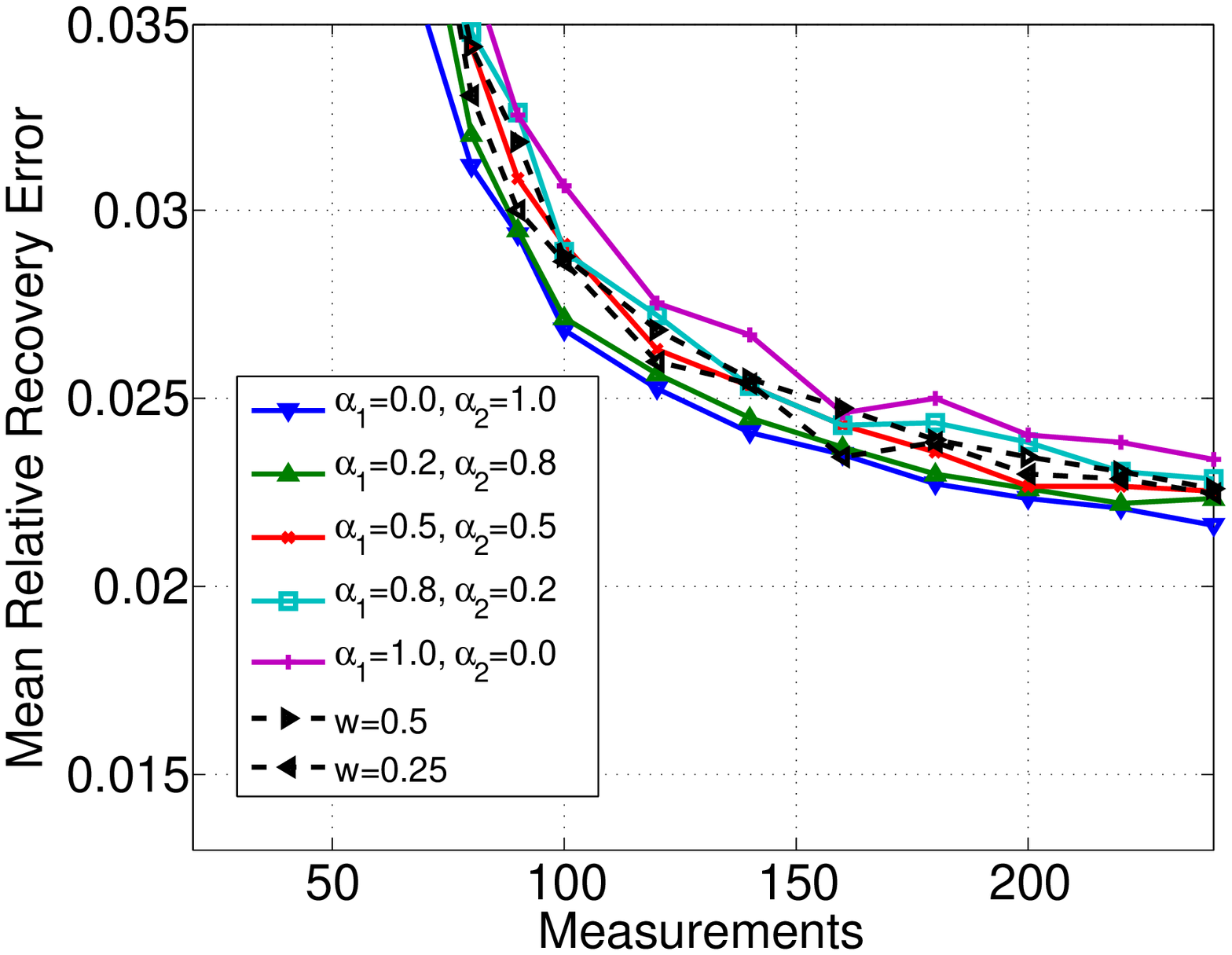} \\
(a) $\alpha=\alpha_1=\alpha_2=1, \rho_1+\rho_2=\rho=1$ & (b) $\alpha = 0.5, \rho = 1, \rho_1=\rho_2 = 0.5$
\end{tabular}
\caption{Comparison of the mean \edit{relative} recovery error over 500 trials versus the number of measurements taken while using weighted $\ell_1$-minimization with a single support estimate and a single weight (black dashed) and two disjoint support estimates with two distinct weights (solid, marker as indicated). We set $w_1=0.5$, $w_2=0.25$ and (a) vary the sizes of $\rho_1$ and $\rho_2$ for fixed $\rho$, $\alpha$, $\alpha_1$, and $\alpha_2$ and (b) vary the sizes of $\alpha_1$ and $\alpha_2$ for fixed $\alpha$, $\rho$, $\rho_1$, and $\rho_2$. }
\label{fig::weightedL1_TwoSupportEstimates}
\end{figure}

In our next next set of experiments, the signal $x$ is synthetically generated with a given signal distribution (i.e., with a specified probability of each entry being nonzero). Again, the signal is of dimension $n=256$ with standard Gaussian nonzero values, the measurement matrix is Gaussian, the measurement noise $z$ is i.i.d. Gaussian with mean zero and standard deviation 0.01, and \edit{100} trials are performed. The optimal choice of weight given the signal distribution is not obvious, and the determination of the optimal relationship is beyond the scope of this paper. We find, however, that the following method empirically performs well for the signal models tested. Let $P_i$ denote the probability that entry $i$ in $x$ is nonzero for $i=1,\dots,n$. Then, we take the weights to be\footnote{The relationship $w_i = 1-P_i$ may seem more natural, however, we found that the more aggressive relationship between $w_i$ and $P_i$ as implemented in the experiments provided superior performance. 
} $w_i = \frac{e^{-5P_i}-e^{-5}}{1-e^{-5}}$ (see Figure \ref{fig::exponential weight}). For comparison with weighted $\ell_1$-minimization with a single weight, all indices $i$ with $P_i > 0$ are assigned the same weight $w$.

Figure \ref{fig::synthetic experiments power} displays the recovery results for signals with a power law distribution. That is, $P_i = \frac{1}{i}$ for $i=1,\dots,n$. However, for any $i$ such that $\frac{1}{i}< 0.025$, we set $P_i=0$ so that the same weight is not applied on \textit{all} indices in the single weight case (if the same weight were applied on all indices, the result would be the same as the standard un-weighted $\ell_1$-minimization in (\ref{L1 min})). 
The probabilities $P_i$ and the non-uniform weights $w_i$ are shown in Figure \ref{fig::synthetic experiments power}. Over all trials, the average signal sparsity was 4.\edit{25}. We compare the non-uniform weight approach to using a uniform weight of $w \in \{0.1, 0.3, 0.5, 0.7, 0.9,1\}$ (note that $w=1$ is equivalent to solving (\ref{L1 min})). We see that although the \edit{relative} recovery error tends to decrease as the weight $w$ decreases in the single weight setting, the non-uniform weight approach outperforms all the others.

Figure \ref{fig::synthetic experiments tree v2} displays the recovery results for signals with a binary tree distribution. That is, the support is organized on a binary tree (plus an extra root node at the top). The first index has just one child; the second and further indices have two children each. An $s$-sparse support is filled by choosing the first index location, and then in each of the $s-1$ remaining rounds, choosing one index randomly among the unselected indices which currently have a selected parent. This type of model is characteristic of natural images (see \cite{DuartWB_Fast,CrousNB_Wavelet} for similar constructions). The probabilities $P_i$ for such a model were calculated experimentally over 10,000 trials of the described tree-sparse support generation, where the sparsity $s=24$ and the dimension $n=256$. Again, for any support index $i$ such that the experimentally calculated probability $P_i<0.025$, we set $P_i=0$. The probabilities $P_i$ and the non-uniform weights $w_i$ are shown in Figure \ref{fig::synthetic experiments tree v2}\footnote{Note that a tree-sparse support is generated using the described method. This means, although unlikely, it is possible for a support to contain indices such that $P_i$ has been set to zero.}.  We again compare the non-uniform weight approach to using a uniform weight of $w \in \{0.1, 0.3, 0.5, 0.7, 0.9,1\}$, and see that the non-uniform weight approach is outperforming all the others. Interestingly, in the single weight setting, the \edit{relative} recovery error decreases as the weight decreases from $w=1$ to $w=0.5$, but then increases for $w=0.3$ and increases even further for $w=0.1$. This illustrates that being overly aggressive in the weight assignment can worsen the reconstruction performance.

\begin{figure}[~htbp]
\centering
\includegraphics[height=1.7in,width=2.1in]{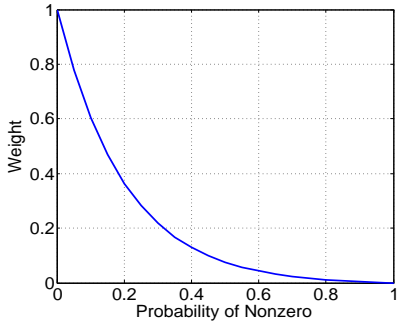} 
\caption{Relationship between probability $P_i$ of  an index being nonzero and non-uniform weight $w_i$ used in our numerical experiments: $w_i = \frac{e^{-5P_i}-e^{-5}}{1-e^{-5}}$. }
\label{fig::exponential weight}
\end{figure}

\begin{figure}[!htbp]
\centering
\begin{tabular}{cc}
\includegraphics[height=2.15in,width=1.85in]{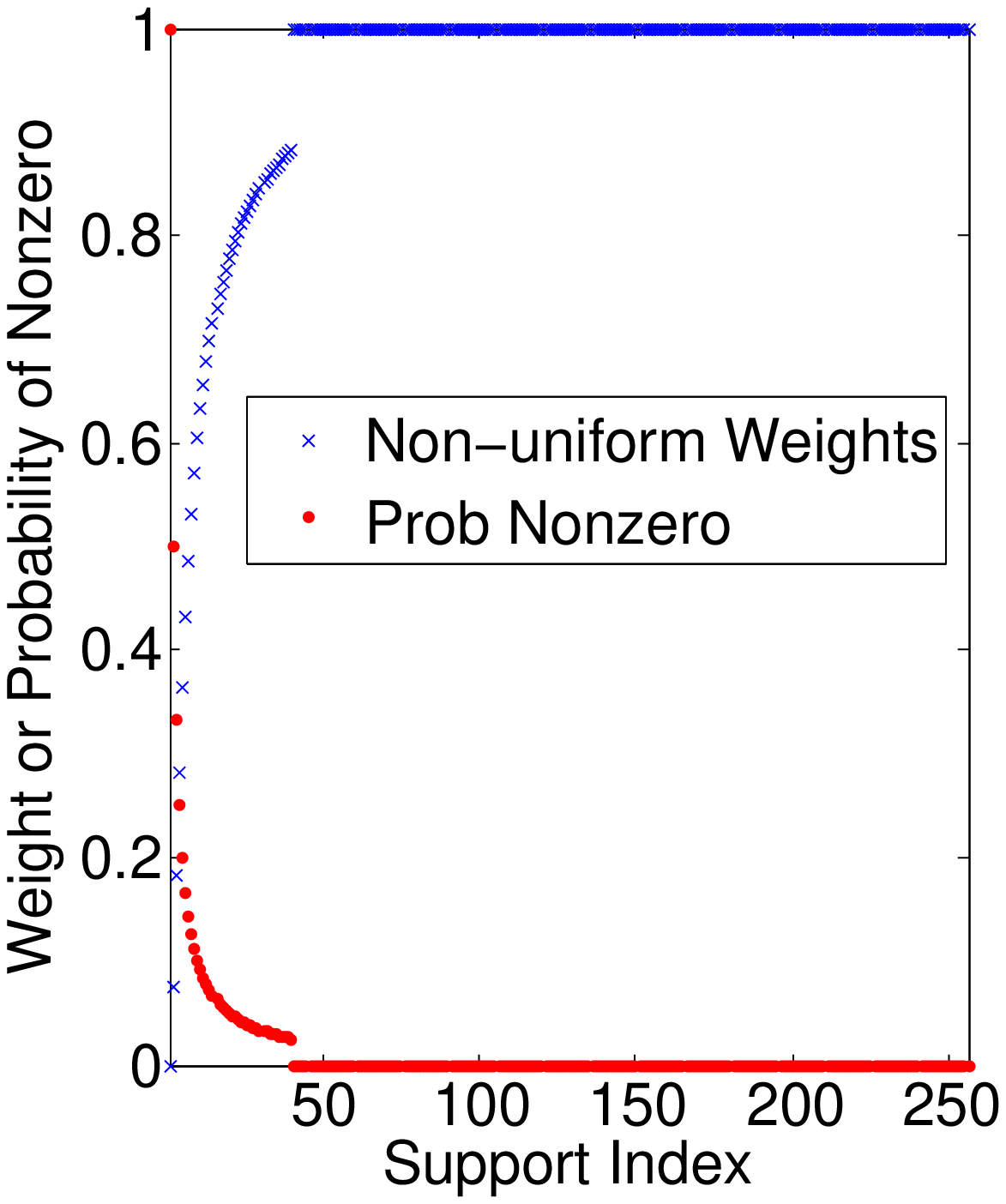} &
\includegraphics[height=2.15in,width=3.25in]{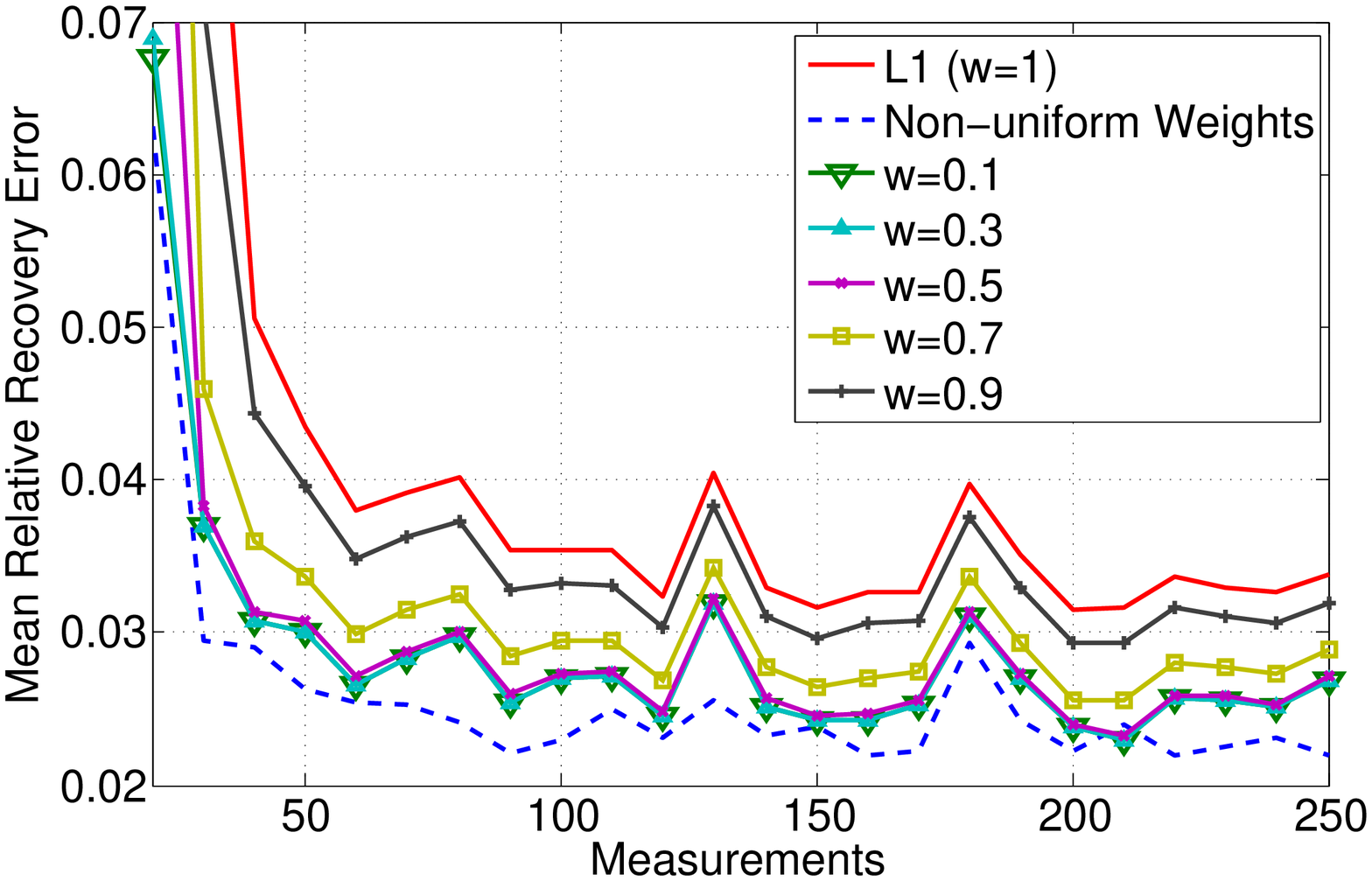} \\
(a) Power Law Distribution & (b) Power Law Recovery
\end{tabular}
\caption{(a) The probability of an index being nonzero (red dot) and the weight assigned (blue x) in the weighted $\ell_1$-minimization with non-uniform weights, where $w_i = \frac{e^{-5P_i}-e^{-5}}{1-e^{-5}}$, versus the support index. (b) Mean \edit{relative} recovery error over \edit{100} trials versus the number of measurements taken for signals with a power law distribution. We compare the mean \edit{relative} recovery error for weighted $\ell_1$-minimization with non-uniform weights (blue dashed), standard $\ell_1$-minimization (red solid), and weighted $\ell_1$-minimization with a single weight (marker as indicated). 
}
\label{fig::synthetic experiments power}
\end{figure}

\begin{figure}[!htbp]
\centering
\begin{tabular}{cc}
\includegraphics[height=2.15in,width=1.85in]{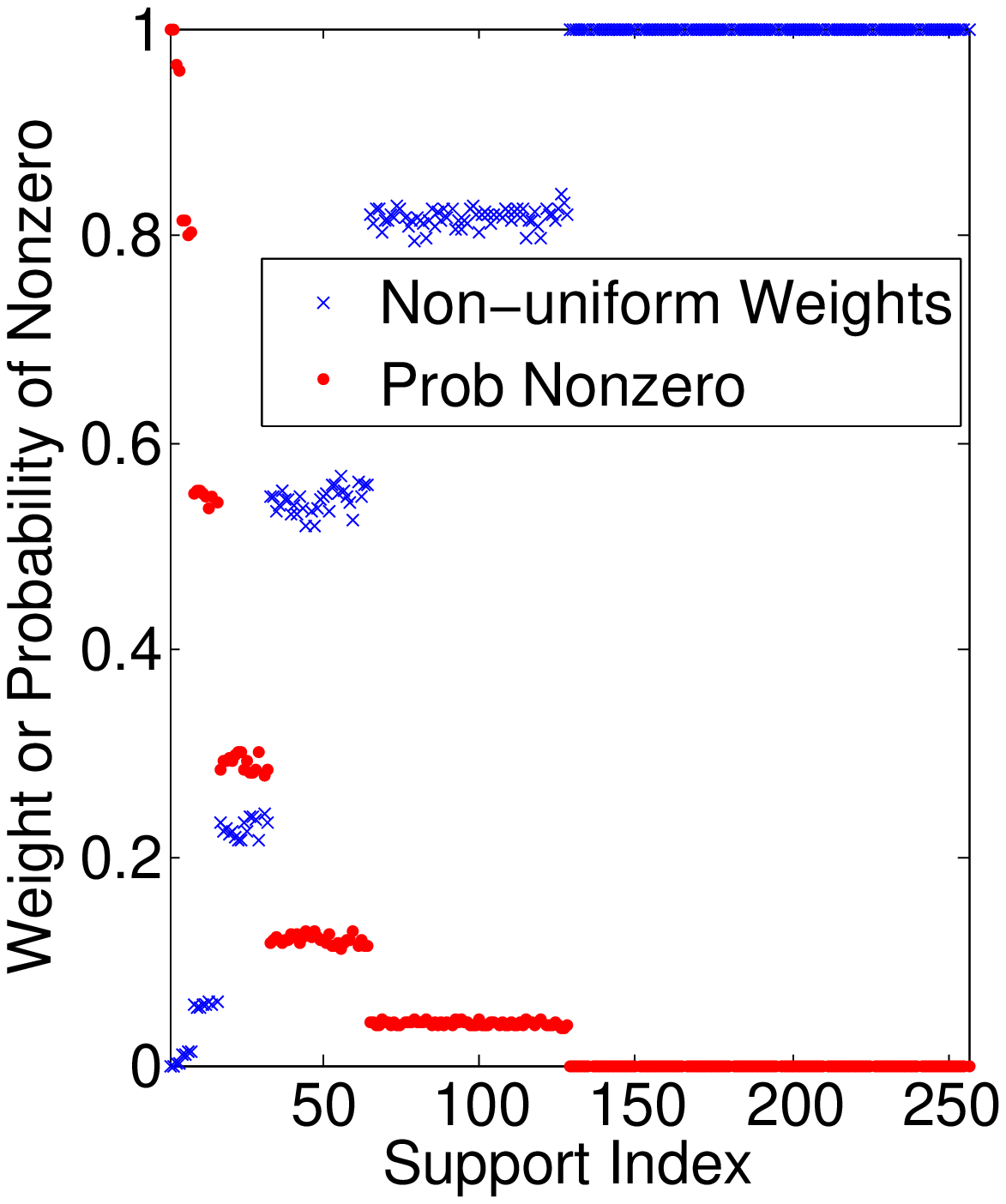} &

\includegraphics[height=2.15in,width=3.25in]{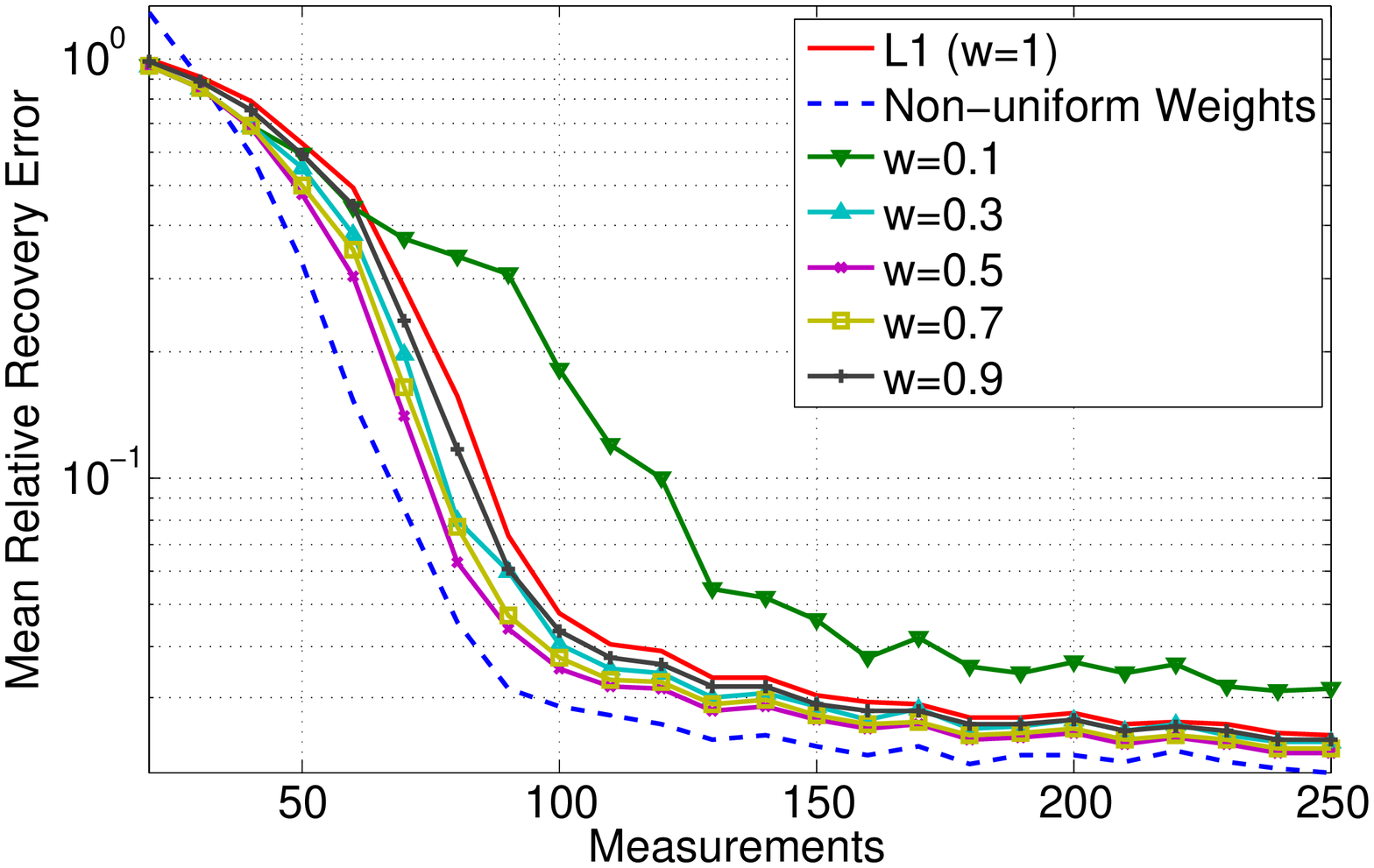} \\
(a) Binary Tree Distribution & (b) Binary Tree Recovery
\end{tabular}
\caption{(a) The probability of an index being nonzero (red dot) and the weight assigned (blue x) in the weighted $\ell_1$-minimization with non-uniform weights, where $w_i = \frac{e^{-5P_i}-e^{-5}}{1-e^{-5}}$, versus the support index. (b) Mean \edit{relative} recovery error over \edit{100} trials versus the number of measurements taken for signals following a binary tree sparsity pattern. We compare the mean \edit{relative} recovery error for weighted $\ell_1$-minimization with non-uniform weights (blue dashed), standard $\ell_1$-minimization (red solid), and weighted $\ell_1$-minimization with a single weight (marker as indicated). 
}
\label{fig::synthetic experiments tree v2}
\end{figure}

\subsection{Application to the Recovery of Video Signals}

One important application where a prior support, or even a prior distribution, can be reasonably estimated is the recovery of video signals since there is often little variation from one frame to the next. In this section, we perform a similar video recovery experiment as in \cite{Saab_weightedL1RIP}, but we also include weighted $\ell_1$-minimization recovery options where the weight does not need to be constant across the entire support estimate. As in \cite{Saab_weightedL1RIP}, we utilize the Foreman sequence at QCIF resolution (i.e., each frame contains $144\times 176$ pixels), and consider only the luma (grayscale) component of the sequence. We split the frames into four blocks  of size $72\times 88$, each of which are processed independently. The measurement matrix for each frame is $A = R D$, where $R$ is an $m \times n$ restriction matrix (i.e., a matrix with $m$ rows from the $n\times n$ identity matrix, selected uniformly at random) with $n=72\times88$, and $D$ is the two-dimensional Discrete Cosine Transform (DCT) (i.e., $D$ is the sparsifying basis). In our experiment, we set $m = \frac{n}{2} = 3168$.

We perform the reconstruction of each frame block using standard $\ell_1$-minimization, weighted $\ell_1$-minimization with a constant weight, and weighted-$\ell_1$ minimization with non-uniform weights. For weighted $\ell_1$-minimization with a constant weight $w$, standard $\ell_1$-minimization is used for the first frame. Then, at frame $j$ for $j\geq 2$ we determine the top 10\% of the DCT coefficients (in magnitude) of the previously recovered frame. Denoting the set of such DCT coefficients by $\widetilde{T}_j$, we set $\widetilde{T} = \bigcup_{i=2}^j \widetilde{T}_j$, and use a constant weight of $w$ on $\widetilde{T}$ in the recovery of frame $j$. Thus, the size of the support estimate $\widetilde{T}$ either grows or remains constant from frame to frame. For weighted-$\ell_1$ minimization with non-uniform weights we follow a similar procedure, but construct an estimated probability of a given DCT coefficient to be in the top 10\% of the DCT coefficients (in magnitude). That is, at frame $j$ for $j\geq 2$ we determine the top 10\% of the DCT coefficients (in magnitude) of the previously recovered frame and set $\widehat{P}_i = \frac{\# \mbox{ of times in top 10\%}}{j-1}$ for each DCT coefficient $i$. Then the weights are taken to be $w_i = \frac{e^{-5\widehat{P}_i}-e^{-5}}{1-e^{-5}}$ (note that for this application, setting $w_i = 1-\widehat{P}_i$ also performed well, but the included relationship between $\widehat{P}_i$ and $w_i$ is even more advantageous). As a final oracle-type comparison, we also include reconstruction results for weighted $\ell_1$-minimization with non-uniform weights when the \textit{true empirical} coefficient probabilities were calculated. For each true frame block, we determine the top 10\% of the DCT coefficients (in magnitude). The empirical probability that coefficient $i$ is nonzero is calculated as $P_i = \frac{\# \mbox{ of times in top 10\%}}{300}$ (note that the Foreman sequence has 300 frames). Then the weights are taken to be $w_i = \frac{e^{-5P_i}-e^{-5}}{1-e^{-5}}$ for \textit{every} frame.  

The reconstruction quality is reported in terms of the peak signal to noise ratio (PSNR) given by the expression
\begin{align}
\mbox{PSNR}(x,\hat{x}) = 10\log_{10} \left(\frac{N\times 255^2}{\|x-\hat{x} \|_2^2} \right),
\end{align}
where $x$ and $\hat{x}$ are the true and estimated full frames, respectively, expressed as a column vector\edit{, and $N$ is the number of pixels in each frame}.
Figure \ref{fig::foreman DCT FINAL} displays the recovery of all 300 frames of the Foreman sequence using weighted $\ell_1$-minimization with both uniform and non-uniform weighting strategies, as well as standard $\ell_1$-minimization. The results demonstrate that a dramatic improvement in PSNR can be achieved when using non-uniform weights, with the recovery method using weights determined from $\widehat{P}_i$ close to the recovery using weights determined from $P_i$. Note that being too aggressive and using a fixed weight of $w=0.2$ eventually results in performance that falls below that of standard $\ell_1$-minimization. For reference, in Figure \ref{fig::foreman DCT FINAL} we also display selected true frames from the Foreman sequence, particularly where we see locally extreme PSNR values. \edit{The drastic improvement in PSNR between frames 177 and 208 for all methods is interesting and perhaps counterintuitive for the weighted schemes. The video sequence transitions from a fairly steady scene to panning across a scene of the sky during these frames, and hence one might conjecture that the prior support estimates would no longer be accurate and the performance would degrade. The improvement, however, is likely because the homogenous sky frames are restricted to the DCT coefficients which have already been identified in the support estimates. The sorted DCT coefficient magnitudes also tend to decrease more rapidly for these frames, explaining why the improved performance is seen for the un-weighted recovery as well. The performance then decreases once a new textured scene is in view.} 

This experiment illustrates a practical unsupervised situation where weighted $\ell_1$-minimization can be used to improve signal recovery while eliminating the need for the practitioner to explicitly choose which weight to use. Here, we have presented an option that determines non-uniform weights on the fly and outperforms simple implementations of weighted $\ell_1$-minimization with a single, constant weight.

\begin{figure}[!htbp]
\centering
\includegraphics[height=3.6in,width=6.1in]{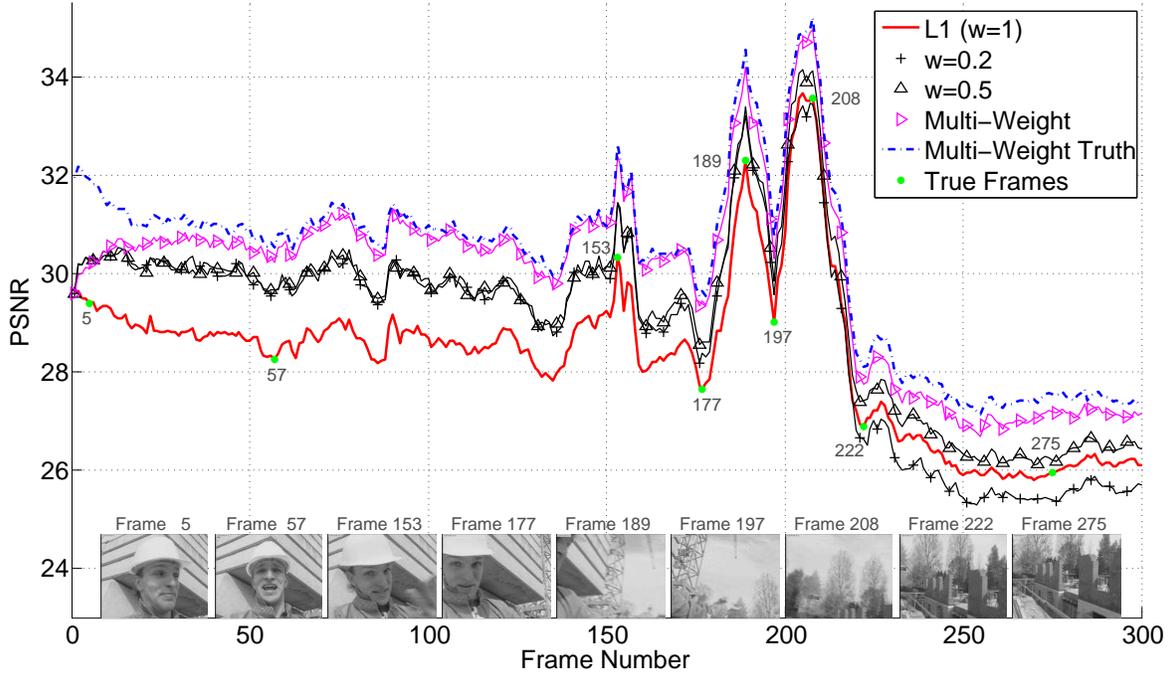} 
\caption{Recovery of frames 1 to 300 of the Foreman sequence at QCIF resolution and processed as four separate blocks of size $72 \times 88$. The sparsity basis is the DCT, the sensing matrix is the restriction operator, and $m = 3168$ measurements are taken. Recovery with standard $\ell_1$-minimization is shown in the solid red curve; recovery with weighted $\ell_1$-minimization with a constant weight in shown in the black curves for $w=0.2$ and $w=0.5$ (marker as indicated); recovery with weighted-$\ell_1$ minimization with non-uniform weights and the estimated $\widehat{P}_i$ is shown in the magenta curve (Multi-Weight, marker as indicated); recovery with weighted-$\ell_1$ minimization with non-uniform weights with the oracle-type $P_i$ is shown in the blue dashed curve (Multi-Weight Truth). Selected true frames are also displayed, with markers indicating where these frames correspond to the standard $\ell_1$-minimization PSNR.}
\label{fig::foreman DCT FINAL}
\end{figure}

\section{Discussion}\label{sec::discussion}
We have generalized the recovery conditions, in terms of the restricted isometry constant of the sensing matrix, of weighted $\ell_1$-minimization for sparse recovery when multiple distinct weights are permitted and arbitrary prior information can be utilized. Our analysis provides both an extension to existing literature that studies weighted $\ell_1$-minimization with a single weight, and an improvement to prior results on weighted $\ell_1$-minimization when multiple distinct weights are allowed. Additionally, we have included simulations that illustrate the theoretical results, and provided examples with synthetic signals and real video data where utilizing many distinct weights is superior to using a single fixed weight. \edit{An interesting extension of this work would be to derive sample complexity bounds for the Gaussian measurement case using a Gaussian width argument similar to that in \cite{SaabM_weighted}.}

\clearpage
\section*{Acknowledgment}
The authors would like to thank Hassan Mansour for sharing his code for the video application and also the reviewers for their thoughtful suggestions which significantly improved the manuscript. {The work of D. Needell and T. Woolf was partially supported by NSF Career DMS-1348721 and the Alfred P. Sloan Foundation. The work of R. Saab was partially supported by a Hellman
Fellowship, and the NSF under grant DMS-1517204.}

\frenchspacing
\bibliographystyle{plain}
\bibliography{bib}

\end{document}